\documentclass[10pt,journal,compsoc]{IEEEtran}
\usepackage{cite}
\usepackage{times}
\usepackage{amsmath}
\usepackage{bbm}
\usepackage{mathrsfs}
\usepackage{amsbsy}
\usepackage{amsxtra}
\usepackage{amsopn}
\usepackage{latexsym}
\usepackage{amsfonts}
\usepackage{epsfig}
\usepackage{graphicx}
\usepackage{amssymb}
\usepackage{stmaryrd}
\usepackage{algorithm}
\usepackage{algorithmic}
\usepackage{color,soul}
\usepackage{epstopdf}
\usepackage{amsthm}

\newtheorem{lemma}{Lemma}
\newtheorem{remark}{Remark}

\newtheorem{theorem}{Theorem}

\usepackage{gensymb}
\usepackage{url}

\usepackage{booktabs} 

\usepackage{array}
\newcolumntype{C}[1]{>{\centering\arraybackslash}m{#1}}
\usepackage{enumitem}
\usepackage[font=footnotesize]{caption}
\begin{document}
\title{Passive Crowd Speed Estimation in Adjacent Regions With Minimal WiFi Sensing}
\author{{ Saandeep Depatla and Yasamin Mostofi 
		\thanks{The authors are with the Department of Electrical and Computer Engineering,
			University of California Santa Barbara, Santa Barbara, CA 93106, USA email:
			$\{$saandeep, ymostofi$\}$@ece.ucsb.edu. This work is funded by NSF CCSS award \# 1611254.}}

}

\IEEEtitleabstractindextext{%
\begin{abstract}
In this paper, we propose a methodology for estimating the crowd speed using WiFi devices without relying on people to carry any device. Our approach not only enables speed estimation in the region where WiFi links are, but also in the adjacent possibly WiFi-free regions. More specifically, we use a pair of WiFi links in one region, whose RSSI measurements are then used to estimate the crowd speed, not only in this region, but also in adjacent WiFi-free regions. We first prove how the cross-correlation and the probability of crossing the two links implicitly carry key information about the pedestrian speeds and develop a mathematical model to relate them to pedestrian speeds. We then validate our approach with 108 experiments, in both indoor and outdoor, where up to 10 people walk in two adjacent areas, with variety of speeds per region, showing that our framework can accurately estimate these speeds with only a pair of WiFi links in one region. For instance, the NMSE over all experiments is 0.18. We also evaluate our framework in a museum-type setting and estimate the popularity of different exhibits. We finally run experiments in an aisle in Costco, estimating key attributes of buyers' behaviors.

\end{abstract}

}

\maketitle
\IEEEdisplaynontitleabstractindextext
\IEEEpeerreviewmaketitle
\IEEEraisesectionheading{\section{Introduction}\label{sec_introduction}}
\IEEEPARstart{C}onsider an area that consists of a number of regions, such as a shopping mall, a retail store, a museum, or a train station.  People may have different average speeds in different regions, depending on the specifics of the regions in terms of popularity, usefulness, or ease of traversing, among other factors.  For instance, one region of a department store can be more popular than other regions, resulting in people slowing down.  A particular exhibit may be less popular in a museum, resulting in people speeding up. Finally, people may slow down in a specific part of a train station due to an ongoing construction work. Thus, the specifics of a particular region can directly affect the speed of the visitors in the corresponding region, as studies have also shown \cite{franvek2013environmental}. In this paper, we are interested in estimating such \textbf{region-dependent speeds}. 
Since a person may not have a constant speed in a region, in this paper \textbf{``speed estimation'' refers to estimating the \textit{average speed} of the people in each region, where the average is the spatial average of the speed of a person in that particular region.}\footnote{We may drop the term ``average'' throughout the paper for brevity.} In other words, people can stop several times in a region, or change their instantaneous speed. We are then interested in estimating their average speed, which is region-dependent and can thus reveal valuable information about the regions.
\begin{figure*}[t]
\centering
\includegraphics[width=0.95\textwidth]{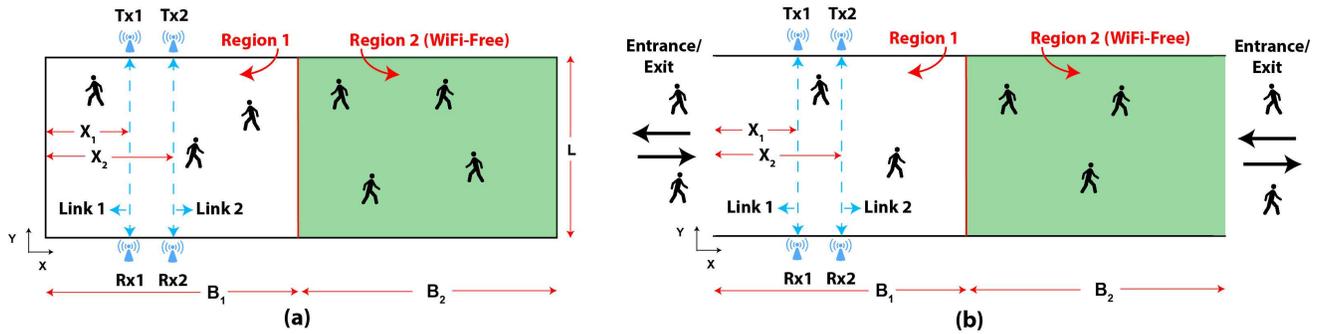}
\vspace{-0.05in}
\caption{Two example scenarios of the problem of interest, where an area consists of two regions, Region 1 and Region 2, as indicated. People move casually throughout the area with a specific speed in each region. A pair of WiFi links are located in Region 1. We are then interested in estimating the region-dependent speeds of both regions, based on only WiFi RSSI measurements of the links and signal availability in Region 1. (a) shows an example of a closed area,
such as an exhibition or a museum, where the total number of people inside the area changes slowly with time and people can traverse back and forth or change directions inside the area any number of times depending on their interest, whereas (b) shows an example of an open area such as a train station, where people can come and go from both regions and can form flow directions.}
\vspace{-0.2in}
\label{fig_illustration}
\end{figure*}
Fig. \ref{fig_illustration} shows two example scenarios of the problem of interest (a closed and an open area), with two adjacent regions. We are interested in estimating the region-dependent speeds of the pedestrians in these two regions, with a pair of WiFi links in only one region (e.g., Region 1 in Fig. \ref{fig_illustration}). More importantly, we are interested in such estimations in large areas. Then, movements of people in Region 2 may not directly affect the links in Region 1. For instance, the WiFi signal may be too weak by the time it gets to Region 2, resulting in a WiFi-free Region 2. As such, we are interested in estimating the crowd speed not only in Region 1 where the links are, but also in the adjacent possibly WiFi-free regions.\textbf{ Estimation of the speeds in both regions, by relying on sensing and WiFi signal availability in only one region, is what we refer to as speed estimation with sensing in only one region in this paper}.\footnote{We emphasize that our approach works the same if the adjacent region is not WiFi-free, or if the movements of people in Region 2 affect the transmitted signals. In other words, our proposed approach does not rely on the availability of the transmitted signals in the adjacent regions and as such can work equally well if the adjacent areas are WiFi-free.} Finally, we are interested in crowd speed estimation without relying on people to carry any device, to which we refer as \textbf{passive} speed estimation.

\textbf{Motivating Examples:} The ubiquity of inexpensive, low-cost, and low-power Internet-of-Things (IoT) sensors 
present great opportunities for learning about our surroundings leading to IoT-enabled smart ambiance. The knowledge of people's walking speed in a particular region can be useful for several applications. For instance, retail stores can learn about the popularity of the products on different aisles, if they know buyers' speeds in different parts of the stores.  Consider an aisle in a retail store containing a specific type of product, for instance. Shoppers that are entering this aisle will walk at a normal pace if the products in the aisle do not attract their attention. On the other hand, they may slow down, or stop to look at the items if they find them of interest. Therefore, by estimating the average speed of the pedestrians in an aisle, the popularity of the products in that aisle can be inferred. This information, in turn, can significantly help with business planning. Similarly, museums can estimate which of their exhibits are more popular, based on the speeds of the visitors. For instance, consider a museum with different exhibits. The visitors typically slow down and spend more time in the exhibit that interests them more. Therefore, by estimating the average speed of the visitors in each exhibit, the popularity of the corresponding exhibit can be inferred. Smart cities can further design the traffic signal timings for the pedestrian crosswalks based on their speeds \cite{laplante2004continuing}.  Furthermore, identifying the slow areas can further help with city planning such as allocation of new roads and facilities, or design of a shopping center. Public places, such as a train station, can further detect abnormal behaviors if an atypical slow down is estimated in a particular area.  Resources can then be allocated accordingly.
\subsection{Related Work}
In this section, we discuss the state-of-the-art for estimating the speed of a crowd.

\textbf{Infrared-Based Approaches: } 
Infrared (IR) sensors can be utilized to sense human activities in an environment. For instance, it has been proposed for counting the total number of people \cite{liu2016occupancy,yang2016review}, or for tracking human motion \cite{yang2014novel}. More recent work has explored classifying the speed of human motion using passive IR sensors. For instance, \cite{yun2014human} classifies the speed of a single person walking in a hallway as slow, moderate, or fast using three IR sensors. A training phase in which a single person walks at different speeds is first utilized to train a classifier, which is then used to classify the speed of a person. This work, however, only considers a single pedestrian. In general, there is no existing IR-based work that can estimate the speed of a crowd of people, or do it with sensing in only one region. More importantly, while IR sensors may be available at the entrance and exit of a retails store, they need to be installed throughout the store for collecting analytics, whereas smart IoT WiFi devices already exist throughout most stores. Nevertheless, we note that the method we propose in this paper can also be implemented with active IR, instead of WiFi, to enable speed estimation of a crowd with IR.

\textbf{Vision-Based Approaches:} Vision-based methods can potentially be used to estimate the speed of pedestrians in the immediate area where the cameras are installed \cite{khanestimating,sourtzinos2011highly,wang2014pedestrian,dridi2014tracking}.  These methods involve using cameras to continuously record a video of the scene in which the pedestrians are walking, followed by computer-vision algorithms to estimate the speeds. However, while consumers are fine with security cameras being probed in an on-demand manner for security purposes, serious privacy concerns arise when cameras are utilized in public places to analyze customer behaviors. For instance, a recent survey on retail shoppers \cite{video_analytics_survey2015} revealed that 75\% of the people who understood the capabilities of vision-based tracking technologies found it intrusive for retails to track their behavior using such a technology. 
Furthermore, employing such tracking techniques could lead to shoppers choosing not to visit the corresponding stores, as reported in \cite{opinionlab_survey}. In summary, vision-based tracking and speed estimation methods have the major drawback of privacy violation. Moreover, vision-based methods involve installing cameras and utilizing complex computer-vision algorithms which can be expensive. For instance, Walmart discontinued its in-store vision-based tracking technology after a few months, as it was too expensive \cite{walmart}. Finally, vision-based techniques can only estimate the speed of people in the areas that are in the direct line-of-sight of the cameras.

Radio Frequency (RF) signals, on the other hand, can alleviate some of the drawbacks associated with the vision-based systems.  For this reason, there has been a considerable interest in using RF signals for estimating some of the characteristics of a pedestrian flow, such as the number of pedestrians in an area \cite{depatla2015occupancy,depatla2018percom,xi2014electronic}, the locations of the people \cite{xu2013scpl}, the walking direction \cite{wu2016widir}, the walking speed \cite{jiang2014communicating}, and other sensing applications \cite{karanam20173d}. In particular, the work on speed estimation, using RF signals, can be classified into device-free passive and device-based active methods, as we summarize next.

\textbf{Device-Based Active RF Approaches:} Device-based active methods depend on the information provided by a mobile  device carried by the pedestrians, such as the Medium Access Control (MAC) data, to track people. However, these methods require the shoppers to carry a wireless device, or an on-body sensor, which limits their applicability. More importantly, if a store is to use shoppers' devices to gather store analytics, it can only gather crude, low resolution tracking data, based on monitoring which router the device is connected to in the store (i.e., this data may not directly translate to speed estimation in different aisles). Even then, serious privacy concerns limit the applicability of such an approach in public places. For instance, Nordstrom, a clothing company which implemented an active WiFi-based in-store tracking technology to analyze the behavior of their customers, withdrew it due to privacy concerns of the shoppers \cite{nordstorm}. Furthermore, a recent survey on active WiFi tracking technology \cite{opinionlab_survey} revealed that $80$\% of the shoppers do not like to be tracked based on their smartphones, while $43$\% do not want to shop at a store that employs active WiFi tracking technology.

\textbf{Device-Free Passive RF Approaches: }The device-free passive methods, on the other hand, leverage the interaction of RF signals with the pedestrians and hence do not require the pedestrians to carry any device. In this manner, they can preserve the privacy. Among the device-free methods, \cite{sigg2014telepathic} classifies the speed of a single person walking in a circle of radius $2\ m$
, based on the RSSI measurements of a mobile phone located at the center of the circle.
A prior training phase, in which RSSI measurements are collected when a single person is walking in the area with three different speeds, is utilized. \cite {shi2014monitoring} classifies the speed of a single person, using FM radio receivers.
Similarly, a training phase in which a person walks at different speeds is used.  However, in these work, only one person is considered in the area and classification of a single speed is performed based on extensive prior training. In realistic scenarios, such as in public places, there will be several pedestrians walking at the same time. In \cite{bocca2014multiple}, RSSI measurements of several WiFi links are used to track up to $4$ people 
walking in the same area. Such an approach can in principle be extended towards speed estimation. However, this and other work on tracking \cite{zhang2013rass} typically have to assume very few people (less than $5$)
. Moreover, in order to estimate region-dependent speeds of a crowd of pedestrians, there is no need to track every individual, as we shall see in this paper. 

In our previous work (conference version of this work \cite{depatla2018secon}), we have shown how to estimate the walking speed of multiple people in a single region (i.e., when people are walking with same speed throughout the region). This is a special case of the scenario considered in this paper where the speed of people is the same in both the regions. In this paper, we then build on our previous work to develop a generalized methodology that can estimate the speed of a crowd in two adjacent regions, where people can walk with different speeds in each region, based on only WiFi sensing in one region.
\subsection{Goals and Contributions}
To the best of our knowledge, passive estimation of the speeds of a crowd in multiple regions, with ubiquitous IoT devices utilizing RF sensing in only one region, has not been explored, which is the main motivation for the proposed work. More specifically, our goal in this paper is to estimate the region-dependent (average) speeds of a crowd of pedestrians in two adjacent regions, without a need for them to carry any wireless device, and by measuring the Received Signal Strength (RSSI) of a pair of WiFi links in only one region.  Our approach enables the estimation of the speed not only in the region where the pair of links are, but also in the adjacent WiFi-free regions as well.  It further shows that it is indeed possible to estimate the motion attributes of a crowd in RF-free zones.  Fig.~\ref{fig_illustration} shows two sample scenarios with two regions, Region 1 and Region 2, and with a region-dependent speed, i.e., people walk with (average) speed of $v_1$ in Region 1 and (average) speed of $v_2$ in Region 2. Two links are installed in Region 1, as can be seen. We are then interested in estimating these region-dependent speeds, based on only the RSSI measurements of the links in Region 1, and without relying on any impact people may have on the links when in Region 2. We next summarize our key contributions:
\begin{itemize}[leftmargin=*]
\item We mathematically characterize the probability of crossing a link, by using a Markov chain modeling and borrowing theories from statistical data analysis. Our results reveal the functional dependency of the probability of crossing on the speeds in both regions. They further indicate how different attributes of the two regions, such as the dimensions of the regions, impact the probability of crossing.
\item We show how the average speeds of the two regions can be estimated using the probability of crossing and the cross-correlation of the two links. To the best of our knowledge, this is the first time that the speeds of a crowd in multiple regions are passively estimated with WiFi.  Moreover, this is the first time the speeds of adjacent WiFi-free regions are estimated. \textbf{It is noteworthy that our approach does not require a training phase where people walk in the area with different speeds beforehand.}
\item We conduct a total of $108$ 
experiments, with up to $10$ people walking in both an indoor and an outdoor area that has two regions, with a variety of speeds per region, and show that our approach can accurately estimate the speeds of pedestrians in the two adjacent regions by using the RSSI measurements of a pair of WiFi links located in one region. For instance, the Normalized Mean Square Error (NMSE) of our speed estimation over all the experiments is $0.18$. Furthermore, the overall classification accuracy, when crowd's speed is categorized as slow, normal, or fast, is $85\%$. Finally, \textbf{the sensing is minimal} in the sense that the number of links per the total size of the area to be monitored is considerably small (e.g., 2 links per $ 14\ m\ \times 4.5 \ m$). 
\item We further validate our framework in a museum setting, where there are two exhibitions each containing very different types of displays. We then estimate the region-dependent average speeds of the invited visitors and thus deduce which exhibit was more popular. We finally run an experiment in an aisle in Costco, estimate key attributes of buyers' motion behaviors, and deduce the interest of the buyers in the products in that aisle.
\end{itemize}
We note that while we showcase our approach with $2$ regions, our approach can be easily extended to speed estimation in $M$ adjacent regions for any $M>2$, with minimal sensing i.e., with sensing in less than $M$ regions. The rest of the paper is organized as following. In Section \ref{sec_prob_formulation}, we discuss the problem setup. In Section \ref{sec_speed_estimation}, we mathematically characterize two key statistics, the probability of crossing and the cross-correlation between a pair of WiFi links, and show how they carry vital information on the speeds of pedestrians in both regions, and present a methodology to estimate these speeds accordingly. In Section \ref{sec_exp_result}, we thoroughly validate our framework with several experiments. 
We conclude in Section \ref{sec_conclusions}.

\begin{section}{Problem Setup}\label{sec_prob_formulation}
Consider the scenario where $N$ pedestrians are walking in an area that consists of two adjacent regions, Region 1 and Region 2, with region-dependent speeds, as shown in Fig.~\ref{fig_illustration}. 
A pair of WiFi links are located in one region, which make RSSI measurements as people walk in the two regions. The goal of this paper is to estimate the speeds of the pedestrians in the two adjacent regions, using the WiFi measurements of the links located in one region. To keep the paper applicable to many scenarios, we consider two possible general cases, as shown in Fig.~\ref{fig_illustration}. The first case (Fig.~\ref{fig_illustration}a), can represent a museum, a conference, or an exhibit-type setting where the total number of people inside the overall area changes slowly with time such that it can be considered constant over a small period of time. People can have any motion behavior in this area and can possibly traverse the area several times back and forth, through different regions, depending on their interest. The second case (Fig.~\ref{fig_illustration}b), on the other hand, captures
the cases where people can enter and exit through both regions, and can form flow directions through the area. Then the total number of people can change rapidly with time and cannot be considered a constant. This case represents scenarios like train stations or a store aisle.

As we show in this paper, the estimation of the region-dependent speeds can be achieved for both cases under the same unifying framework. We assume that, $N$, the total number of people in the area (or $N_\textnormal{avg}$, the average number of people for time-varying cases such as Fig.~\ref{fig_illustration}b) is known.
Assuming the knowledge of the total number of people in the area is reasonable for many applications. For instance, in stores, there may be mechanisms (such as door sensors) to count the total number of people in the store.  Then, it would be of interest to estimate the speed of shoppers in different regions.  We further note that the total number of people can also be estimated with additional sensing in the area.  Thus, in this paper we focus on estimating the region-dependent speeds, assuming $N$ (or $N_\textnormal{avg}$), and based on minimal sensing in only one region. 
In this section, we summarize a simple motion model for the pedestrians and briefly discuss their impact on the links. This is then followed by our proposed methodology for estimating the region-dependent speeds in the next section.
\subsection{Pedestrian Motion Model}
In this paper, we assume that people move casually in the two adjacent regions and do not assume any specific pattern for their motion. To describe a casual motion, we then use the simple mathematical model of \cite{depatla2015occupancy}, which we briefly summarize next. Consider the motion of a single person in the workspace of Fig.~\ref{fig_illustration}. Let $x(k)$, $y(k)$, and $\theta(k)$ denote the position along x-axis, the position along y-axis, and the heading of the person w.r.t. the x-axis, at time $k$, respectively. Since the person walks casually in the area, he/she keeps walking in a particular direction, while occasionally changing the direction of motion. This can be captured by using the following model for the heading direction:
\begin{equation}\label{eq_motion_heading}
\theta(k+1) =
\begin{cases}
\theta(k)\hspace{-0.03in} &\textnormal{\ with \ probability\ } p \\
\textnormal{Uniformly in\ } \mu \hspace{-0.03in} &\textnormal{\ with \ probability\ }1-p
\end{cases}
\end{equation}
where $\mu = [-\theta_{\textnormal{max}},\ \theta_{\textnormal{max}}] \cup [\pi-\theta_{\textnormal{max}},\ \pi+\theta_{\textnormal{max}}]$, for the case of Fig. \ref{fig_illustration}a since people can change their direction any time and can traverse the area back and forth as many times as they wish. $\theta_{\textnormal{max}}$ then defines the maximum angle for the direction of motion. 
For instance, when $\theta_{\textnormal{max}} = 90 \degree$, the person can choose any direction in $[0, 2\pi]$. Then, $\theta_{\textnormal{max}}$ allows us to model the motion depending on the environment and scenario. For instance,  $\theta_{\textnormal{max}}$ is typically less than $90\degree$ in long hallways  \cite{weidmann2014pedestrian}. For the case of  Fig.~\ref{fig_illustration}b, we assume that people mainly travel in a forward direction. Thus, we take $\mu = [-\theta_{\textnormal{max}},\ \theta_{\textnormal{max}}]  \textnormal{\ or\ } \mu =[\pi-\theta_{\textnormal{max}},\ \pi+\theta_{\textnormal{max}}]$ depending on the direction of motion.

Based on Eq.~(\ref{eq_motion_heading}), the position dynamics are then given as follows:
\begin{equation}\label{eq_motion_xpos}
x(k+1)=
\begin{cases}
 x(k)+v_1\delta t\ \textnormal{cos}(\theta(k)) &\textnormal{\ if\ } 0 \leq x(k) < B_1 \\
 x(k)+v_2\delta t\ \textnormal{cos}(\theta(k)) &\textnormal{\ if\ } B_1 \leq x(k) < B
\end{cases},
\end{equation}
\begin{equation}\label{eq_motion_ypos}
y(k+1)=
\begin{cases}
 y(k)+v_1\delta t\ \textnormal{sin}(\theta(k)) &\textnormal{\ if\ } 0 \leq x(k) < B_1 \\
 y(k)+v_2\delta t\ \textnormal{sin}(\theta(k)) &\textnormal{\ if\ } B_1 \leq x(k) < B
\end{cases},
\end{equation}
where $\delta t$ is the time step, and $B=B_1+B_2$. For the case of Fig.~\ref{fig_illustration}a, we assume that when a person encounters any of the four boundaries of the area, he/she reflects off of the boundary, similar to a ray of light.\footnote{This boundary behavior is only assumed for the purpose of modeling. In our experiments, we have no control over how people walk.}  For the open area of Fig.~\ref{fig_illustration}b, on the other hand, we assume a mainly forward flow from each entrance towards the opposite exit. Then, the person exits the area upon reaching the opposite exit. 
We then use this motion model in the next section when developing our methodology for estimating the speeds.
\subsection{Effect of Pedestrians on the WiFi Signals}
As shown in Fig.~\ref{fig_illustration}, a pair of WiFi links located in Region 1, make wireless measurements as the pedestrians walk in the two regions. When a pedestrian (or multiple) crosses a link, the corresponding RSSI measurement will drop, to which we refer as Line of Sight (LOS) blockage.  When people do not block the LOS but they are in the vicinity of a link, they can still impact the received signal through multipath. 
The proposed methodology of this paper is based on utilizing only  the LOS blockage impact. In Section \ref{sec_init_data_process}, we show how to estimate the LOS blockage sequence from the received RSSI measurements.
\end{section}
\begin{section}{Estimation of pedestrian speeds}\label{sec_speed_estimation}
In this section, we propose a framework to estimate the region-dependent speed of pedestrians in two adjacent regions, using a pair of WiFi links located in only one region, as shown in Fig.~\ref{fig_illustration}. More specifically, we first derive a mathematical expression for the probability of pedestrians crossing a WiFi link. We then characterize the cross-correlation between the two links. Our analysis shows that these parameters carry key information on the speeds of the pedestrians in both regions, which we then use to estimate the speeds.  A key feature of our approach is that it only relies on WiFi signal availability in the region where the links are but can deduce the speed of the crowd in the adjacent possibly WiFi-free region. In this section, we first characterize the probability of crossing and the cross-correlation for the case of the closed area of Fig.~\ref{fig_illustration}a, the analysis of which is more involved since a person can reverse the direction of motion anytime and can bounce
back and forth in the area as many times as he/she wishes. We then show how to extend the analysis to the case of open area  of Fig.~\ref{fig_illustration}b, putting everything under one unifying umbrella.

\subsection{Probability of Crossing a Link}\label{sec_prob_cross}
Consider Fig.~\ref{fig_illustration}a and the motion model of Eq.~(\ref{eq_motion_heading})-(\ref{eq_motion_ypos}). Since the heading, and the positions along the x-axis and y-axis at time $k+1$, depend only on the corresponding values at time $k$, we use a Markov chain model to describe the motion dynamics of each pedestrian. We then use the properties of the corresponding Markov chain to mathematically derive the probability of crossing a given link by a single pedestrian and show its dependency on the speeds of each region. This is then followed by characterizing the probability that any number of people cross a given link. We note that the probability of crossing problem of interest to this section is considerably different from that of \cite{depatla2015occupancy}, since there are two regions with links in only one region.  As such, a new characterization and methodology is required as we develop in this section.

For the purpose of modeling, we discretize the work-space and assume that people can choose only discrete positions along x-axis, y-axis, and the heading direction.\footnote{This is only for the purpose of mathematical characterization. In practice, the positions and heading of the pedestrians are naturally not limited to these discrete values.} More specifically, $\theta(k) \in \mu^d = \{-\theta_{\textnormal{max}},\ -\theta_{\textnormal{max}}+\Delta \theta,\ \cdots,\  \theta_{\textnormal{max}}\} \cup \{\pi-\theta_{\textnormal{max}},\ \pi-\theta_{\textnormal{max}}+\Delta \theta,\ \cdots, \pi+\theta_{\textnormal{max}}\}$, $x(k) \in \{0,\ \Delta x,\ \cdots,\ B_1+B_2\}$, and $y(k) \in \{0,\ \Delta y,\ \cdots,\ L\}$, where $\Delta \theta $, $ \Delta x$, and $\Delta y $ denote the discretization step size for heading and position along x-axis and y-axis respectively. Let $N_\theta$ denote the number of discrete angles for the heading. Furthermore, let $N_1$ and $N_2$ represent the number of discrete positions along the x-axis in Region 1 and Region 2 respectively.

Let $\Theta(k)$ denote the random variable representing the heading of a pedestrian at time $k$. Let $\pi^\theta(k)$ represent the corresponding probability vector with the $i^\textnormal{th}$ element $(\pi^\theta(k))_i  = \textnormal{Prob}(\Theta(k) = (\mu^d)_i)$, where Prob(.) is the probability of the argument, and $(\mu^d)_i$ denotes the $i^\textnormal{th}$ element of the set $\mu^d$. Then from Eq.~(\ref{eq_motion_heading}), we have the following Markov chain for the heading $\Theta(k)$:
\begin{equation}
\pi^\theta(k+1) = \pi^\theta(k) P^\Theta ,
\end{equation}
where $P^\Theta$ is the probability transition matrix for the heading with $(P^\Theta)_{ij} = \textnormal{Prob}(\Theta(k+1)=(\mu^d)_j|\Theta(k)=(\mu^d)_i)$ and is given by $(P^\Theta)_{ij}  = \delta(i-j) \times p + \frac{1-p}{N_\theta} =  (P^\Theta)_{ji}$,
where $\delta(.)$ is the Dirac-delta function, $N_\theta = \textnormal{card}(\mu^d)$, and card(.) denotes the number of elements in the argument. Since the probability transition matrix $P^\Theta$ is symmetric, it is a doubly-stochastic matrix, which implies a uniform stationary distribution for $\Theta(k)$ \cite{meyer2000matrix}. This implies that the probability that a pedestrian heads in any given direction (in $\mu^d$) is the same asymptotically.

Let $X(k)$ denote the random variable representing the position of a pedestrian along the x-axis at time $k$. Similar to the heading direction, we can describe the dynamics of $X(k)$ using a Markov chain. Let $P^X$ denote the corresponding probability transition matrix for $X(k)$. We then have the following lemma for the stationary distribution of $X(k)$.
\begin{lemma}
The stationary distribution of $X(k)$ is given by $\gamma = [c_1 {e_1}\  c_2 {e_2}]$, where $c_1$, $c_2$ are constants, and ${e_1}$, ${e_2}$ are $N_1$ and $N_2$-dimensional row-vectors with all their elements as $1$.
\end{lemma}
\begin{proof}
Let $P^X$ be partitioned as $P^X = \begin{bmatrix}
P_{11}  & P_{12}   \\
P_{21}  & P_{22}
\end{bmatrix},
$
where $P_{11}$ is a square matrix of dimension $N_1$. Further, $P_{ij}, \textnormal{\ for\ } \ i,j \in \{1,2\}$, specify the transition probabilities from positions in Region \textit{i} to positions in Region \textit{j}. The stationary distribution of the partitioned transition matrix $P^X$ is shown in \cite{meyer1989stochastic} to be $\gamma = [k_1 {\gamma_1}\  k_2 {\gamma_2}] $, where $k_1$ and $k_2$ are constants, and $\gamma_1$ and $\gamma_2$ are the stationary distribution vectors corresponding to the probability transition matrices, $S_{11}$ and $S_{22}$, defined as follows:
\begin{equation} \label{eq_stoc_comp}
\begin{split}
S_{11} & = P_{11} + P_{12}(I_{N_2}-P_{22})^{-1}P_{21} \\
S_{22} & = P_{22} + P_{21}(I_{N_1}-P_{11})^{-1}P_{12},
\end{split}
\end{equation}
where $I_{N_1}$ and $I_{N_2}$ are the identity matrices of dimensions $N_1$ and $N_2$ respectively.

Consider any two positions, $r \Delta x$ and $q \Delta x$, along the x-axis that are in the same region (i.e., with the same speed).  Then, based on \cite{depatla2015occupancy},
$
{\textnormal{Prob}(r \Delta x \rightarrow q \Delta x)} = {\textnormal{Prob} ( q \Delta x \rightarrow r \Delta x)},
$
where ${\textnormal{Prob}(r \Delta x \rightarrow q \Delta x)}$ denotes the probability of going from $q \Delta x$ to $r \Delta x$ in one time step. Since the speed of the pedestrians is the same within a region, we then have,
\begin{equation} \label{eq_prop1}
\begin{split}
P_{11} = &P_{11}^T \textnormal{\ and\ }
P_{22} = P_{22}^T.
\end{split}
\end{equation}
Furthermore, by choosing the step size $\Delta x$ such that $q \Delta x$ can be reached from $r \Delta x$ in one time step if and only if $|q-r| \leq 1 $, we have the following property for $P_{12}$ and $P_{21}$.
\begin{equation}\label{eq_prop2}
\begin{split}
(P_{12})_{ij} & \neq  0 \textnormal{\ iff\ } i=N_1, j= N_1+1 \\
(P_{21})_{ij} & \neq  0 \textnormal{\ iff\ } i=N_1+1, j= N_1.
\end{split}
\end{equation}

By substituting Eq.~(\ref{eq_prop1}) and (\ref{eq_prop2}) in (\ref{eq_stoc_comp}), we get, $S_{11}  = S_{11}^T \textnormal{\ and \ }
S_{22}  = S_{22}^T$.
Since $S_{11}$ and $S_{22}$ are symmetric, the corresponding stationary distributions are uniform, implying $\gamma_1 = \frac{e_1}{N_1}$, and $\gamma_2 = \frac{e_2}{N_2}$. Therefore, the stationary distribution of $P^X$ is $ \gamma = [c_1 {e_1}\  c_2{e_2}]$, where $c_1 = \frac{k_1}{N_1}$ and $c_2 = \frac{k_2}{N_2}$ are constants. This proves the lemma.
\end{proof}

Lemma $1$ states that the position of a pedestrian along the x-axis has a uniform asymptotic distribution within each region.

We next derive the probability that a pedestrian crosses a link, given that the pedestrian is in a region where there is a link (Region 1 in this case). We then use this conditional probability of crossing the link, along with Lemma $1$, to derive the overall probability of crossing. We first mathematically define crossing/blocking a link. We say that a pedestrian crosses/blocks a given link\footnote{In this paper, we consider WiFi links that are located parallel to the y-axis (see Fig.~\ref{fig_illustration}). However, the derivation of the probability of crossing can be extended to any general link configuration following a similar approach.} located at $X_i$ along the x-axis, at time $k+1$, if either $x(k+1)  \geq {X_i} \textnormal{\ and\ } x(k) \leq {X_i} $ or $x(k+1)  \leq {X_i} \textnormal{\ and\ } x(k) \geq {X_i} $.
With this definition for the cross/block, we then have the following lemma for the conditional probability of crossing a given link, given that the pedestrian is in the region where there is a link.
\begin{lemma}
Given that a person is in Region 1, the probability of crossing a given link in Region 1 is given by $p_{c|1} = \frac{v_1 \delta t\ \textnormal{sinc}(\theta_\textnormal{{max}})}{B_1}$, where $\textnormal{sinc}(\theta_\textnormal{{max}}) \triangleq \frac{\textnormal{sin}(\theta_\textnormal{{max}})}{\theta_\textnormal{{max}}}$ with $\theta_\textnormal{{max}}$ in radians.
\end{lemma}
\begin{proof}
Consider a link located in Region 1 of Fig.~\ref{fig_illustration}a, whose x-coordinate is $X_i$. 
$X_i$, for instance, can represent $X_1$ or $X_2$ of Fig.~\ref{fig_illustration}a. Let the position of the person at time $k$ be $x(k) \leq X_i$. The person crosses the link at time $k+1$, if he/she chooses a direction $\theta(k)$ at time $k$ such that $
x(k)+ v_1\delta t \textnormal{cos}(\theta(k)) \geq X_i, \textnormal{\ which results in \ } |\theta(k)| \leq \textnormal{cos}^{-1}\Big(\frac{X_i - x(k)}{v_1 \delta t}\Big)
$, where $|.|$ is the absolute value of the argument. Since $|\theta(k)| \leq \theta_\textnormal{max}$, in order to cross the link, the heading direction should be as follows:
\begin{equation}
|\theta(k)| \leq \textnormal{min}\Big\{\theta_\textnormal{max},\ \textnormal{cos}^{-1}\Big(\frac{X_i - x(k)}{v_1 \delta t}\Big)\Big\}.
\end{equation}
Since the heading direction is uniformly distributed over $\mu^d$, the probability that a person at $x(k)$ crosses the link in Region 1 at time $k+1$, $p_{c|1}^{x(k)}$, is given by,
\begin{equation}\label{eq_prob_cross_given_pos}
p_{c|1}^{x(k)}=\frac{\textnormal{min}\Big\{\theta_\textnormal{max},\ \textnormal{cos}^{-1}\Big(\frac{X_i - x(k)}{v_1 \delta t}\Big)\Big\}}{2\theta_\textnormal{max}},\ \textnormal{for\ } x(k) \leq X_i.
\end{equation}
By symmetry, it can be seen that $p_{c|1}^{x(k)}$, for $x(k) \geq X_i$, is given by,
\begin{equation}\label{eq_prob_cross_given_pos2}
\begin{split}
p_{c|1}^{x(k)}=\frac{\textnormal{min}\Big\{\pi-\theta_\textnormal{max},\ \pi-\textnormal{cos}^{-1}\Big(\frac{x(k) - X_i}{v_1 \delta t}\Big)\Big\}}{2\theta_\textnormal{max}},\  \\ 
\end{split}
\end{equation}
The probability of crossing the link given the person is in Region 1, $p_{c|1}$, is then obtained by summing over all the positions in Region 1 from which a cross can occur:
\begin{equation}\label{eq_prob_cross_given_region}
p_{c|1} = \sum_{x(k)=X_i-v_1\delta t}^{X_i+v_1\delta t} \frac{\Delta x}{B_1} p_{c|1}^{x(k)},
\end{equation}
where $\frac{\Delta x}{B_1}$ is the probability that a pedestrian is located at any given position in Region 1.
By substituting Eq.~(\ref{eq_prob_cross_given_pos}) and (\ref{eq_prob_cross_given_pos2}) in (\ref{eq_prob_cross_given_region}) and letting $\delta t \rightarrow 0$, we get,
\begin{equation}\label{eq_cond_prob_cross}
p_{c|1}  =
\frac{1}{2B_1 \theta_\textnormal{max}}  \int_{X_i-v_1\delta t}^{X_i+v_1\delta t}\hspace{-2mm} \textnormal{min}\Big\{\theta_\textnormal{max},\ \textnormal{cos}^{-1}\Big(\Big|\frac{X_i - x(k)}{v_1 \delta t}\Big |\Big)\Big\} dx.
\end{equation}
By simplifying Eq.~(\ref{eq_cond_prob_cross}) further, we get
\begin{equation}\label{eq_cond_prob_cross_final_form}
p_{c|1} = \frac{v_1 \delta t \textnormal{sin}(\theta_\textnormal{max})}{B_1\theta_\textnormal{max}},
\end{equation}
which proves the lemma.
\end{proof}

By using Lemma $1$ and Lemma $2$, we then have the following theorem for the probability of crossing a given link by a single pedestrian.
\begin{theorem}
The probability of crossing a given link by a single pedestrian, $p_{c,\textnormal{single person}}$, walking with the speed $v_1$ in Region 1 and speed $v_2$ in Region 2, is given by, $p_{c, \textnormal{single person}} = \frac{v_1 v_2 \delta t\textnormal{sinc}(\theta_\textnormal{max})}{v_1B_2+v_2B_1} $.
\end{theorem}
\begin{proof}
The probability of crossing a given link in Region 1 by a single pedestrian is given by,
\begin{equation}
p_{c, \textnormal{single person}}  = c_1\ p_{c|1},
\end{equation}
where $c_1$, defined in Lemma $1$, denotes the probability of the pedestrian being in Region 1, and $p_{c|1}$ is the conditional probability that the pedestrian crosses the given link in Region 1, if he/she is in Region 1.

To find the probability $c_1$, we use the pseudo-aggregation properties of the underlying Markov chain \cite{rubino2014markov}. More specifically, for the transition matrix $P^X$, defined in Lemma $1$, with a stationary distribution of the form $[c_1e_1\ c_2e_2] $, the constants $c_1$ and $c_2$ are given by the stationary distribution of the probability transition matrix $P$, as we show next.
\begin{equation}\label{eq_aggregated_matrix}
P= \begin{bmatrix}
p_{11}\ p_{12} \\
p_{21}\ p_{22}
\end{bmatrix},
\end{equation}
where
$p_{ij}=\frac{\mathbf{1}^T P_{ij}\mathbf{1}}{N_i}, \textnormal{\ for\ } i,j \in \{1,2\}$,
and $\mathbf{1}$ denotes a column vector whose elements are all $1$. We can then prove that the stationary distribution of $P$ in Eq.~(\ref{eq_aggregated_matrix}) is $(c_1,c_2)$ \cite{rubino2014markov}.

It can be seen that $p_{12}$, is the probability of crossing from Region 1 to Region 2. From Lemma $2$, we have,
\begin{equation}\label{eq_aggre_matrix_probs}
\begin{split}
p_{12} & = \frac{p_{c|1}}{2} \textnormal{\ and \ }
p_{21}  = \frac{p_{c|2}}{2}.
\end{split}
\end{equation}
By substituting Eq.~(\ref{eq_aggre_matrix_probs}) in (\ref{eq_aggregated_matrix}) and solving for the stationary distribution of $P$, we have,
$c_1  = \frac{v_2 B_1}{v_1 B_2 + v_2 B_1}, \textnormal{\ and\ }
c_2  = \frac{v_1 B_2}{v_1 B_2 + v_2 B_1}$.
The probability of crossing a given link in Region 1 by a single pedestrian can then be characterized as follows,
\begin{equation}\label{eq_prob_cross_1per}
p_{c, \textnormal{single person}}  = c_1\ p_{c|1} = \frac{v_1 v_2 \delta t\  \textnormal{sinc}(\theta_\textnormal{max})}{v_1B_2+v_2B_1}.
\end{equation}
This proves the theorem.
\end{proof}
\begin{remark}
\normalfont Note that if there was a link in Region 2, the probability of a single pedestrian crossing it would have been the same.  This can be seen from the expression for $p_{c, \textnormal{single person}}$ by interchanging $B_1$ with $B_2$ and $v_1$ with $v_2$. Further, note that the probability of crossing is independent of the location of the link within Region 1.
\end{remark}
Since there are $N$ people walking in the area, we next characterize the probability that any number of people cross a given link, $p_c(v_1,v_2)$, assuming that pedestrians' motions are independent. We then have the following for the closed case:
\begin{equation}\label{eq_prob_cross_nppl}
p_c(v_1,v_2) = 1 - (1-p_\textnormal{c, single person})^N.
\end{equation}
From Eq.~(\ref{eq_prob_cross_nppl}), it can be seen that the probability of any number of pedestrians crossing the link is a function of the speeds of the pedestrians in both regions. Furthermore, from Remark 1, we can see that the probability of crossing a link in Region 2, if there was one in Region 2, will not provide any additional information in terms of the speeds in Region 1 and 2, as it has the same exact function form as the probability of crossing a link in Region 1. In other words, it would not have been possible to estimate the speeds by utilizing two links, one in Region 1 and one in Region 2.
\subsection{Characterizing the Cross-correlation}\label{sec_cross_corr}
As shown in Fig.~\ref{fig_illustration}, a pair of WiFi links in one region (Region 1) make wireless measurements as people walk in two adjacent regions. We next characterize the cross-correlation between these two links and show how it carries vital information on the speeds. 

Consider the closed area of Fig. \ref{fig_illustration}a. We say an event $E_l$ happens at a link, if $l >0$ number of people block the link. Let $Y_1(k)$ and $Y_2(k)$ denote the event sequences corresponding to Link 1 and Link 2, as defined below:
\begin{equation*}
Y_i(k) = \begin{cases}
 l &\textnormal{\ if $E_l$ happens at time\ } k \\
 0  &\textnormal{\ otherwise}
\end{cases},
\textnormal{\ for \ }  i \in \{1,2\}.
\end{equation*}
In this section, we show that the cross-correlation between the event sequences of the two links carry key information about the speeds of the pedestrians. We show how to estimate the event sequences from real data in the next section.

The cross-correlation between the two event sequences, $Y_1(k)$ and $Y_2(k)$, is given by
\begin{equation}\label{eq_cross_corr}
R_{Y_1Y_2}(\tau,v_1,v_2)=\frac{\textnormal{Cov}\Big(Y_1(k), Y_2(k+\tau)\Big)}{\sqrt{\textnormal{Var}\Big(Y_1(k)\Big)\textnormal{Var}\Big(Y_2(k+\tau)\Big)}},
\end{equation}
where $\textnormal{Cov(. , .)}$, and $\textnormal{Var(.)}$ denote the covariance and variance of the arguments, respectively. Since the pedestrians walk independent of each other, we have,
\begin{equation}\label{eq_event_superpos_x}
Y_i(k) = \sum_{j=1}^N Y_i^j(k), \textnormal{\ for\ } i \in \{1,2\},
\end{equation}
where
$Y_i^j(k) = 1$  if $j^{\textnormal{th}}$ person blocks Link $i$ at time $k$, 0 otherwise for $i \in \{1,2\}$.
Since we assume independent motion for the pedestrians, it can be easily confirmed that the numerator and the denominator of Eq.~(\ref{eq_cross_corr}) are proportional to $N$, and therefore the cross-correlation becomes independent of $N$. This can be seen by substituting Eq.~(\ref{eq_event_superpos_x}) in (\ref{eq_cross_corr}), and further simplifications, which results in
\begin{equation}
\begin{split}
 R_{Y_1Y_2}&(\tau,v_1,v_2) = \\&
\frac{\textnormal{Prob}(Y_2^j(k+\tau) =1|Y_1^j(k)=1) - p_{c,\textnormal{single person}}}{ 1- p_{c,\textnormal{single person}}}, \\
& \textnormal{\ for\  any } j \in  \{1,2,\cdots,N\}.
\label{eq_cross_corr_final_form}
\end{split}
\end{equation}
It can be seen that Eq.~(\ref{eq_cross_corr_final_form}) is independent of $N$. For the case of open area, since the number of people in the area changes with time, $N$ should be considered a random variable. Then, by assuming that peoples' arrival into the area follow a Poisson process, by substituting Eq. (\ref{eq_event_superpos_x}) in (\ref{eq_cross_corr}), and after further simplification, we get an expression similar to Eq. (\ref{eq_cross_corr_final_form}), which is a function of only the motion dynamics of a single pedestrian.    

While it is considerably challenging to derive a closed-form expression for the cross-correlation, the dependency on the speeds can be easily seen. For instance, the first term in the numerator of Eq.~(\ref{eq_cross_corr_final_form}),$\textnormal{\ Prob}(Y_2^j(k+\tau) =1|Y_1^j(k)=1)$, is the probability that the $j^\textnormal{th}$ person is at Link 2 at time $k+\tau$, given that he/she is at Link 1 at time $k$. Clearly this depends on the speeds at which the $j^\textnormal{th}$ person is walking in both regions. Hence the cross-correlation in Eq.~(\ref{eq_cross_corr_final_form}) contains information about the speeds.  However, given the vicinity of the two links, and by considering all the possible motion patterns of the people, it can be easily seen that the cross-correlation carries more information on the speed of Region 1, as compared to Region 2. As such, in the next part, we utilize it for the estimation of the speed in Region 1, as we shall see.
\end{section}
\subsection{Speed Estimation for the Closed Area} \label{sec_speed_est_closed_area}
As shown in Sections \ref{sec_prob_cross} and \ref{sec_cross_corr}, the probability of crossing a WiFi link, and the cross-correlation between the two WiFi links, carry key information about the speeds of the pedestrians in the two adjacent regions. Equations (\ref{eq_prob_cross_nppl}) and (\ref{eq_cross_corr_final_form}) further model these relationships, which we then use to estimate the speeds of the pedestrians in the two regions.

Let $Y_1^{\textnormal{exp}}$ and $Y_2^{\textnormal{exp}}$ denote the event sequences, corresponding to the two WiFi links, obtained from an experiment. Let $R_{Y_1,Y_2}^{\textnormal{exp}}(\tau)$ denote the cross-correlation between the event sequences $Y_1^{\textnormal{exp}}$ and $Y_2^{\textnormal{exp}}$, and let $p_{c,1}^{\textnormal{exp}}$, $p_{c,2}^{\textnormal{exp}}$ denote the probability of crossing Link 1 and Link 2 respectively. The probability of crossing can be computed from the event sequences as follows:
\begin{equation}
\begin{split}
p_{c,i}^{\textnormal{exp}}= \frac{\delta t}{T} \times {\textnormal{Number of events in\ } Y_i^{\textnormal{exp}}}, \textnormal{\ for\ } i \in \{1,2\},
\end{split}
\end{equation}
where $T$ denotes the total time for which the data is collected, and $\delta t$ is the discretization step size.

Since the cross-correlation of Eq.~(\ref{eq_cross_corr_final_form}) is independent of the total number of people, $N$, we first estimate $v_1$ from the cross-correlation without assuming the knowledge of $N$. Then, given $N$ and an estimate of the speed in Region 1, i.e., $\widehat{v_1}$, we use the probability of crossing in Eq.~(\ref{eq_prob_cross_nppl}) to estimate the speed in Region 2. More specifically, we have,
\begin{equation}\label{eq_speed_est}
\begin{split}
\widehat{v_1} = & \min_{v_1,v_2} \sum_{\tau = 0}^{\tau = T} \Big(R_{Y_1,Y_2}^{\textnormal{exp}}(\tau) - R_{Y_1,Y_2}(\tau,v_1,v_2)\Big)^2 \\
\widehat{v_2} = & \min_{v_2} \Big( p_c^{\textnormal{exp}} - p_c({\widehat{v_1},v_2})\Big)^2,
\end{split}
\end{equation}
where $p_c^{\textnormal{exp}} = \frac{p_{c,1}^{\textnormal{exp}}+p_{c,2}^{\textnormal{exp}}}{2}$. In other words, given that each link will have the same probability of cross, we average the experimental probability of crossing of the two links in order to reduce the impact of errors. We further only estimate $v_1$ from the cross-correlation, since it is heavily dependent on $v_1$, as discussed earlier. As for evaluating $R_{Y_1,Y_2}(\tau,v_1,v_2)$, we utilize simulations, which are low complexity since the cross-correlation is independent of $N$ and can thus be simulated  for only one person. More specifically, for any given speed pair, we simulate one person walking in the area and generate the event sequences corresponding to the two links in the area. $R_{Y_1,Y_2}(\tau,v_1,v_2)$ is then obtained by computing the cross-correlation between the two event sequences.  Finally, the parameter $\theta_{\textnormal{max}} $ in $p_c(v_1, v_2)$ is assumed to be $45\degree$ in all our results of the closed areas in the next section since they involve long hallways. We note that our results are not very sensitive to this choice of $\theta_{\textnormal{max}}$, and $ \theta_{\textnormal{max}}$ for a wide range of angles near $45\degree$ will lend similar results as we shall see in the next section.
\subsection{Speed Estimation for the Open Area}
Consider the open area scenario shown in Fig.~\ref{fig_illustration}b. The number of people in the area can change during the sensing period and should be considered a random variable. However, as explained in Section III-C, since the cross-correlation is not a function of the number of people, Eq. (\ref{eq_speed_est}) can still be used to estimate the speed $v_1$. We next show how to characterize the probability of crossing for the open area in order to estimate $v_2$.

Let $\lambda$ denote the rate of arrival of people into the area (from both regions). We assume that the rate of departure of people from the area is also $\lambda$. This will be the case as long as the average number of people, $N_\textnormal{avg}$, averaged over a small time interval, does not change significantly with time. Furthermore, we assume that each person mainly has a forward flow, i.e., she/he mainly walks in a forward direction and rarely turns back. The probability of crossing a link is then related to the rate of arrival as follows:
\begin{equation}
\begin{split}
p_c(v_1,v_2) = &\textnormal{Number of events in time interval [0 T]}  \times \frac{\delta t}{T} \\
= & \lambda \delta t, \label{eq_rate}
\end{split}
\end{equation}
To relate $p_c$ to the average speed of people in the two regions, we next use a theory from queuing systems.

Consider the overall area as a queuing system in which every person is serviced until the person exits. Then, the Little's law of queuing theory \cite{little2011or} relates the average number of people being serviced, $N_\textnormal{avg}$,  to the average time spent in the area by a person, $T_\textnormal{avg}$, and the rate of arrival, $\lambda$,  as follows:
\begin{equation}\label{eq_avg_number}
N_\textnormal{avg} = \lambda T_\textnormal{avg}.
\end{equation}
Since we assume that people mainly walk in a forward direction, the average time spent in the corridor can be approximated as follows:\footnote{We note that a better approximation of the average time can be calculated by considering the motion model of people in Section \ref{sec_prob_formulation}, as part of our future work.}
\begin{equation}\label{eq_avg_time}
T_\textnormal{avg} \approx \frac{B_1}{v_1}+\frac{B_2}{v_2}.
\end{equation}
From Eq. (\ref{eq_rate}), (\ref{eq_avg_number}), and (\ref{eq_avg_time}), we can characterize the probability of crossing in terms of the speeds of people in the two regions as follows:
\begin{equation}\label{eq_prob_cross}
p_c(v_1,v_2) \approx \frac{N_\textnormal{avg} v_1 v_2}{v_1B_2+v_2B_1}  \delta t.
\end{equation}
$v_1$ and $v_2$ can then be estimated by substituting Eq. (\ref{eq_prob_cross}) in Eq. (\ref{eq_speed_est}).
\begin{remark}
 Consider the expression derived for $p_c$ of Eq. (\ref{eq_prob_cross_nppl}), for the closed case.  If we assume that the probability of simultaneous crosses are negligible, we can approximate Eq. (\ref{eq_prob_cross_nppl}) with $\frac{N v_1 v_2 \delta t\textnormal{sinc}(\theta_\textnormal{max})}{v_1B_2+v_2B_1} $.  For the open case, Eq. (\ref{eq_avg_time}) becomes a better approximation if $\theta_\textnormal{max}$ is small. Then, by approximating
 $\theta_\textnormal{max} \approx 0$, we then have the probability of crossing of the closed
 case approximated by $\frac{N v_1 v_2 \delta t}{v_1B_2+v_2B_1} $, which is similar to the
 expression derived for the open case in Eq. (\ref{eq_prob_cross}). As mentioned earlier, Eq. (\ref{eq_avg_time}) can be more rigorously related to $\theta_\textnormal{max}$ as part of future work.
\end{remark}
\begin{section}{Experimental Results}\label{sec_exp_result}
In this section, we validate the proposed methodology of Section \ref{sec_speed_estimation} with several experiments. We start with a number of experiments in closed areas in both indoor and outdoor, where different number of people walk in two adjacent regions, with a variety of possible speeds per region, and show that our framework can estimate the speeds with a good accuracy. We then run experiments in a museum-style setting, where two exhibitions showcase two very different types of displays. Our approach can then accurately estimate the visitor speeds in both exhibits, and thus deduce which exhibit is more popular. We finally test our framework in an open aisle of a retail store, Costco, and estimate the rate of arrival and speed of people in the aisle, thus inferring the interest of people in the products of the aisle. We next start by explaining the experimental setup and the initial data processing.
\begin{figure}[t]
\centering
\includegraphics[width=0.95\linewidth]{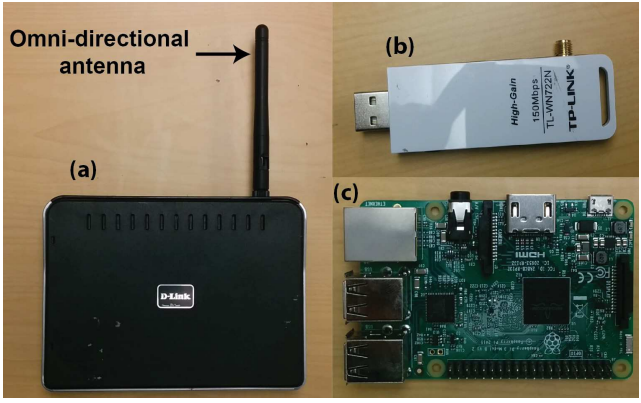}
\caption{(a) D-Link WBR 1310 wireless router along with an omni-directional antenna, (b) the TP-Link wireless N150 WLAN card, (c) Raspberry Pi board used to control the data collection process and synchronize the two WiFi links.}\label{fig_hardware}
\end{figure}
\subsection{Experiment Setup}\label{sec_exp_setup}
As shown in Fig.~\ref{fig_illustration}, our experiments consist of pedestrians walking in two adjacent regions with different possible speeds in each region. A pair of WiFi links located in one region make wireless RSSI measurements to estimate the speeds of pedestrians in both regions. 
We use a D-link WBR-1310 WiFi router that operates in 802.11g mode as a Tx node and a TP-Link Wireless N150 WLAN card configured to operate in 802.11g mode as a Rx node for each link. In order to receive and store the wireless measurements, the WLAN cards need to be interfaced with a computer via a USB connection. We use a portable credit card-sized computer, Raspberry Pi (RPI), for this purpose. Furthermore, to transmit and receive the wireless signals, we use omni-directional antennas at both the router and the WLAN card of each link. Fig.~\ref{fig_hardware} shows the WiFi router, WLAN card, RPI, and the omni-directional antenna used in our experiments.

In order to derive the cross-correlation from the experimental data, the receivers of the two WiFi links need to be synchronized in time. To achieve this, we interface the Rx nodes of both WiFi links to the same RPI and program them to receive the wireless signals at the same time instants from their corresponding transmitters. The data is collected at a rate of 20 samples/second at each receiver of the WiFi link. Since the two WiFi links are located at a close distance (of the order of meters), each link is configured to operate in a different sub-channel of the 2.4 GHz wireless band to avoid any interference. Specifically, we use sub-channel 1, which operates at $2.41$ GHz for one link, and sub-channel 11, which operates at $2.47 $ GHz, for the other link. This separates the two links by the widest frequency margin in the 2.4 GHz WiFi band. 
Fig.~\ref{fig_experiment_scenario} (left) and Fig.~\ref{fig_experiment_scenario} (right) show the resulting experimental setup in an outdoor and an indoor area respectively.
\subsubsection{Experimental Speeds}
As shown in Fig.~\ref{fig_illustration}, our experiments involve pedestrians walking at various speeds in each of the two adjacent regions. In our experiments of Section \ref{sec_exp_result_sub}, we ask people to walk casually throughout the area containing two regions, maintaining a specific speed $v_1$ in Region 1 and $v_2$ in Region 2. We consider three speeds, $0.3\ m/s$, $0.8\ m/s$, and $1.6\ m/s$, for each region, which results in $9$ possible combinations for the speeds in the two regions.  To help people walk at the correct speeds, we make use of a mobile application called ``Frequency Sound Generator" which generates an audible tone every second. Each person then listens to this application on his/her mobile and takes a step of length $v_1$, while walking in Region 1 and a step of length $v_2$, while walking in Region 2, every time he/she hears the tone. This ensures correct speeds for people walking in each region. In order to take steps of length $v_1$ in Region 1 and $v_2$ in Region 2, we have people practice their step lengths to match $v_1$ and $v_2$ prior to the experiments. This procedure is employed only to ensure an accurate ground-truth of speeds in each region, which is used in assessing the performance of our approach. In our museum-type experiments and the experiments in the aisle of Costco, the speeds of people are naturally determined by their interests in each region, and as such there is no control over peoples' speeds in those experiments. 
\subsection{Initial Data Processing}\label{sec_init_data_process}
As shown in Section \ref{sec_speed_estimation}, our framework is based on the event sequences of a pair of WiFi links located in one region, with the events corresponding to people crossing a WiFi link. Therefore, we need to first extract the event sequences of each WiFi link from the corresponding RSSI measurements. We next describe this process.
\begin{figure*}[t]
\centering
\begin{minipage}{1.0\textwidth}
\begin{center}
\includegraphics[width=1\linewidth]{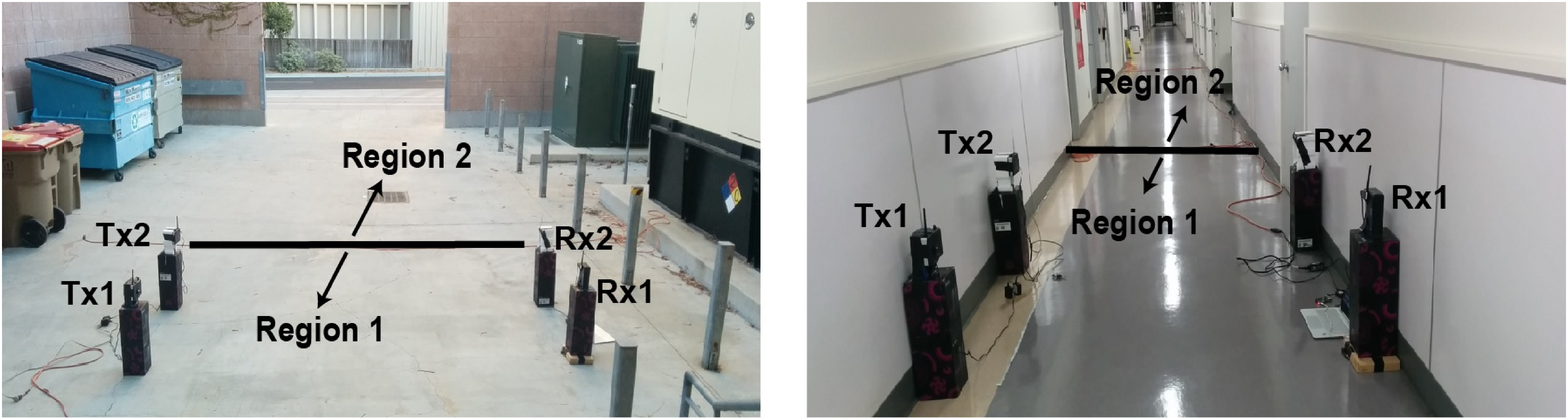}
\caption{(left) The outdoor area of interest and (right) the indoor area of interest. Each area is divided into two regions, Region 1 and Region 2, as separated by the black line in both outdoor and indoor cases. The dimensions of the outdoor area are $L=4.26\  m $, $B_1= 5.5\ m$, $B_2= 8.8\ m $, and that of indoor area are $L= 2.25\ m $, $B_1=7\ m $, $B_2=13\ m$ (see Fig.~\ref{fig_illustration} for definitions of $B_1$ and $B_2$). Two WiFi links, each consisting of a transmitter and a receiver are located in Region 1.}
\label{fig_experiment_scenario}
\end{center}
\end{minipage}
\end{figure*}

To convert the RSSI measurements into an event sequence, we first identify all the dips in the RSSI measurements and the associated times at which the dips occur. Let  $k_i, \textnormal{\ for\ }i \in \{1,2,\cdots,I\}$, denote these times, and  let $Z(k_i)$ denote the corresponding RSSI measurement at time $k_i$. 
The event sequence, $Y_{i}^{\textnormal{exp}}(k)$, is then obtained from the RSSI measurements as follows:
\begin{equation*}\label{eq_thresholding}
\begin{split}
Y_{i}^{\textnormal{exp}}(k)=
\begin{cases}
l & \textnormal{ if } k = k_i \textnormal{\ and\ } Z(k_i) \textnormal{\ is closest to } R_{l,i} \\
0\ & \textnormal{otherwise}
\end{cases}, \\
\textnormal{\ for\ } i \in \{1,2\},
\end{split}
\end{equation*}
where $R_{l,i}$ denotes the RSSI measurement of the $i^{\textnormal{th}}$ WiFi link when $l$ people simultaneously block the $i^{\textnormal{th}}$ link. We find the values of $R_{l,i}$ by performing a small calibration phase in which $l$ (up to 2) people simultaneously block the $i^{\textnormal{th}}$ WiFi link and the corresponding RSSI is measured.\footnote{We need to collect this only for small $l$ as the probability of $l$ people simultaneously blocking the LOS link is negligible for higher $l$.} Note that small variations in $R_{l,i}$ due to factors such as different dimensions of people crossing the WiFi link have a negligible impact on our results. For instance, we collect $R_{l,i}$ data for only $2$ people in the calibration phase, while a total of $10$ different people walk in each campus experiment.



\subsection{Experimental Validations and Discussions}\label{sec_exp_result_sub}
In this Section, we extensively validate our framework by estimating the speeds of people in two adjacent regions of an area using the aforementioned experimental setup.



\begin{figure*}[h!]
\centering
\begin{center}
\includegraphics[width=1\linewidth]{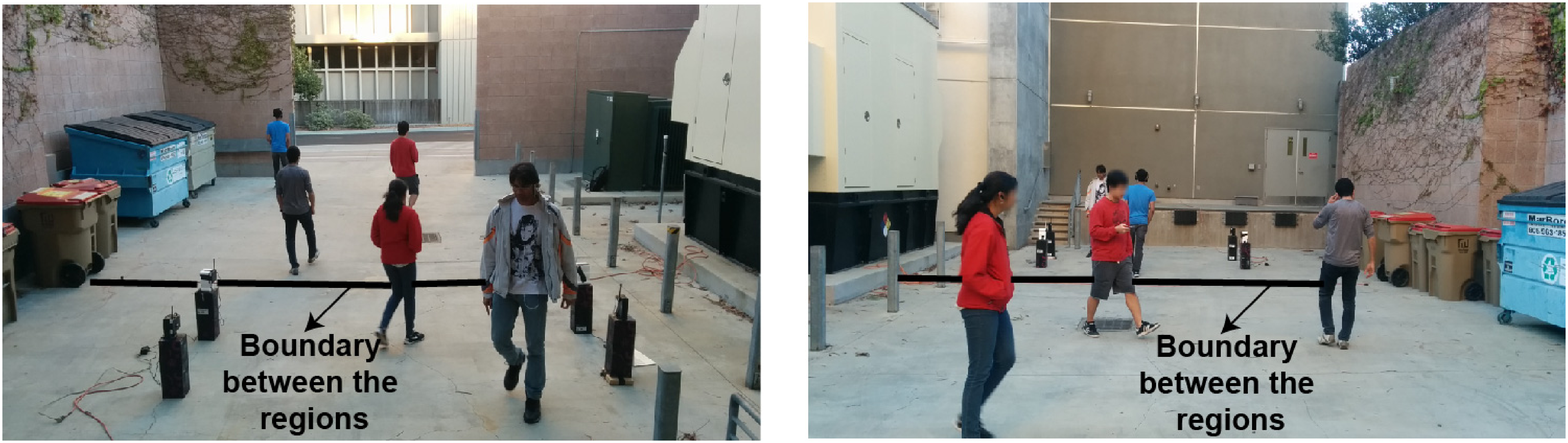}
\caption{The outdoor area of interest with two snapshots of people walking in the area. The black line separates the area into two regions. People move casually throughout the area with the given region-specific speed. A pair of WiFi links located in Region 1 makes wireless measurements to estimate the speed of people in both regions.}
\label{fig_outdoor_scenario}
\end{center}
\end{figure*}
\begin{table}[t]
\begin{center}
\renewcommand{\arraystretch}{1.1}%
\begin{tabular}{|C{2cm}|C{3cm}|}\hline
\shortstack{\\ \bf True speeds \\ ($v_1$, $v_2$)} & \shortstack{ \\ \bf Estimated speeds \\ ($\widehat{v_1}$, $\widehat{v_2}$) } \\ \hline
(0.8, 0.8) & (0.9, 0.9)   \\ \hline
(0.8, 0.3) & (0.8, 0.3)  \\ \hline
(0.8, 1.6) & (0.8, 2.3)  \\ \hline
(0.3, 0.8) & (0.4, 0.9)  \\ \hline
(0.3, 0.3) & (0.4, 0.4)  \\ \hline
(0.3, 1.6) & (0.3, 2.4) \\ \hline
(1.6, 0.8) & (1.7, 0.6)  \\ \hline
(1.6, 0.3) & (1.8, 0.5)  \\ \hline
(1.6, 1.6) & (1.9, 2)  \\ \hline
\end{tabular}
\caption{A sample performance of our speed estimation approach for Region 1 ($v_1$) and Region 2 ($v_2$) of the outdoor area of Fig.~\ref{fig_outdoor_scenario} and the case of $N=5$ people.}\label{tab_sample_result_outdoor}
\end{center}
\end{table}
Fig.~\ref{fig_experiment_scenario} (left) and (right) show the considered outdoor and indoor closed areas of interest respectively. Each area is divided into two regions, with a pair of WiFi links located in one of the regions. The dimensions of the outdoor area are 
$L=4.26\  m $, $B_1= 5.5\ m$, $B_2= 8.8\ m $, $X_1 =2.5\ m $,  $X_2 = 3.7\ m $, while the dimensions of the indoor area are $L= 2.25\ m $, $B_1=7\ m $, $B_2=13\ m $, $X_1 =2.5 \ m $, $X_2 =4\ m $ (see Fig.~\ref{fig_illustration}a). People are then asked to walk casually throughout the area, with a specific region-dependent speed.  Sample snapshots of people walking in the outdoor and indoor areas are shown in Fig.~\ref{fig_outdoor_scenario} and Fig.~\ref{fig_indoor_scenario} respectively.
\begin{table}[t]
\begin{center}
\renewcommand{\arraystretch}{1.1}%
\begin{tabular}{|C{2cm}|C{3cm}|}\hline
\shortstack{\\ \bf True speeds \\ ($v_1$, $v_2$)} & \shortstack{ \\ \bf Estimated speeds \\ ($\widehat{v_1}$, $\widehat{v_2}$)} \\ \hline
(0.8, 0.8)  & (0.9, 0.9)   \\ \hline
(0.8, 0.3)  & (1, 0.5)   \\ \hline
(0.8, 1.6)  &(0.9, 1.6)   \\ \hline
(0.3, 0.8) &(0.5, 0.9)   \\ \hline
(0.3, 0.3)  &(0.5, 0.3)   \\ \hline
(0.3, 1.6) &(0.4, 1.9)   \\ \hline
(1.6, 0.8)  &(1.9, 0.7)   \\ \hline
(1.6, 0.3)  &(1.7, 0.4)   \\ \hline
(1.6, 1.6)  &(1.9, 2.1)   \\ \hline
\end{tabular}
\caption{A sample performance of our speed estimation approach for Region 1 ($v_1$) and Region 2 ($v_2$) of the indoor area of Fig.~\ref{fig_indoor_scenario} and the case of $N=9$ people.}\label{tab_sample_result_indoor}
\end{center}
\end{table}
\begin{table*}
\begin{minipage}[b]{0.3\textwidth}
\renewcommand{\arraystretch}{1.2}%
\begin{tabular}{|c|c|c|c|}
\hline
\bf Speed   & $v_1$ & $v_2$ & $v_1$ or $v_2$ \\ \hline
\bf NMSE & 0.11 & 0.24 & 0.18   \\ \hline
\end{tabular}
\caption{NMSE of the estimation of speeds in each region as well as the overall NMSE of the speeds in any of the two regions.}\label{tab_nmse_all}
\end{minipage}
\hspace{0.09 in}
\begin{minipage}[b]{0.3\textwidth}
\renewcommand{\arraystretch}{1.2}%
\begin{tabular}{|c|c|c|c|}
\cline{2-4}
\multicolumn{1}{c|}{} & \multicolumn{3}{c|}{\bf NMSE} \\
\hline
\bf Scenario & $v_1$  & $v_2$ & $v_1$ or $v_2$   \\ \hline
\bf Outdoor & 0.09 & 0.16  & 0.12  \\ \hline
\bf Indoor & 0.14  & 0.33  & 0.23\\ \hline
\end{tabular}
\caption{NMSE of speed estimation for both indoor and outdoor.}\label{tab_nmse_region}	
\end{minipage}
\hspace{0.15 in}
\begin{minipage}[b]{0.33\textwidth}
\renewcommand{\arraystretch}{1.2}%
\begin{tabular}{|C{1.5cm}|C{0.7cm}|C{0.7cm}|C{1cm}|}
\cline{2-4}
\multicolumn{1}{c|}{} & \multicolumn{3}{c|}{\bf NMSE} \\
\hline
\shortstack{\\ \bf Number of \\ \bf people} &$v_1$  & $v_2$ & $v_1$ or $v_2$   \\ \hline
\bf N=5 & 0.06 & 0.20  & 0.13  \\ \hline
\bf N=9 & 0.16  & 0.29  & 0.23\\ \hline
\end{tabular}
\caption{NMSE of speed estimation based on the total number of people walking in the area.}\label{tab_nmse_people}
\end{minipage}
\end{table*}
\begin{figure*}[t]
\centering
\begin{center}
\includegraphics[width=0.999\linewidth]{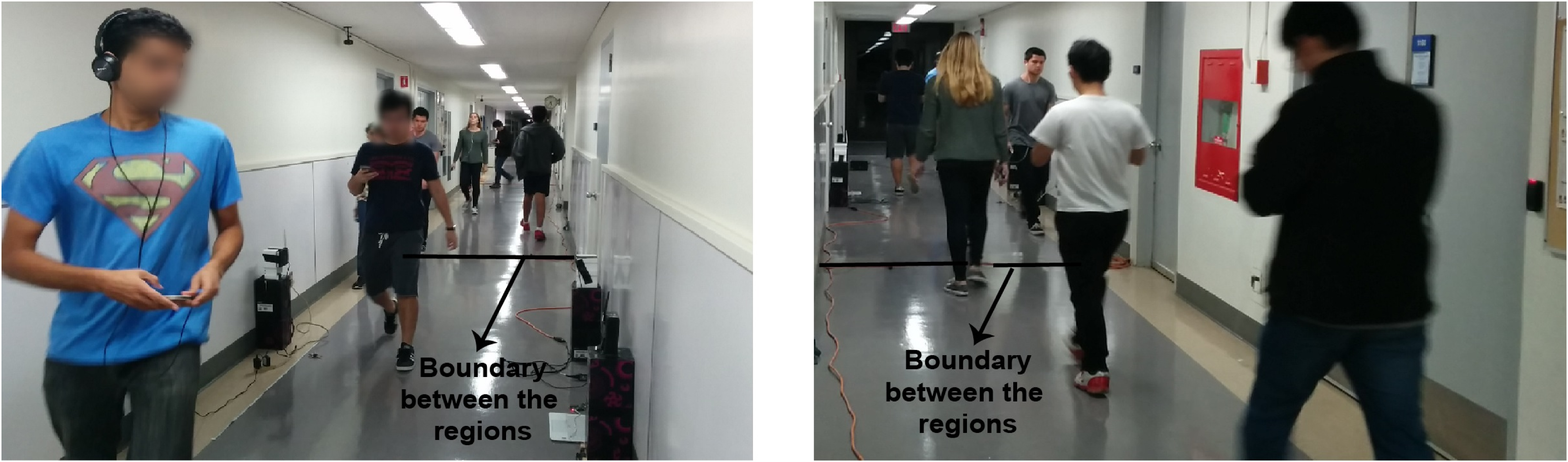}
\caption{The indoor area of interest with two snapshots of people walking in the area. The black line separates the area into two regions. People move casually throughout the area with the given region-specific speed. A pair of WiFi links located in Region 1 makes wireless measurements to estimate the speed of people in both regions.}
\label{fig_indoor_scenario}
\end{center}
\end{figure*}
We have conducted several experiments in these areas with different number of people walking at a variety of speeds. More specifically, we test the proposed methodology with $9$ possible combinations of speeds for $(v_1, v_2)$ for the two adjacent regions. For each pair of speeds, we then run a number of experiments with both $5$ and $9$ people walking in the area. For any given speed, people are instructed on how to walk with that specific speed as discussed in Section \ref{sec_exp_setup}.
Table \ref{tab_sample_result_outdoor} shows a sample performance of our approach when $5$ people are walking in the outdoor area and for all the 9 speed combinations, while Table \ref{tab_sample_result_indoor} shows a sample performance when $9$ people are walking in the indoor area. It can be seen that our proposed methodology can estimate the speeds of people in the adjacent regions with a good accuracy, for both indoor and outdoor cases, by using a pair of WiFi links located in only one region. 

To further validate our framework statistically, we repeat each speed pair 3 times, on different days, for both cases of $5$ and $9$ people walking in the area. This amounts to 108 overall sets of experiments. To evaluate the performance, we calculate the NMSE. 
Table \ref{tab_nmse_all} shows the overall NMSE of the estimation error for speed of Region 1 as 0.11, for speed of Region 2 as 0.24, and for the speed in any of the two regions as 0.18, confirming a good performance. Fig.~\ref{fig_cdf_plot_all} further shows the Cumulative Distribution Function (CDF) of the Normalized Square Error (NSE) for the speed of Region 1, Region 2, and the speed in any region. It can be seen that the NSE is less than $0.15$, $90 \%$ of the time for $v_1$ and $70 \% $ of the time for $v_2$, further confirming a good performance.  We note that the estimation of $v_1$, i.e., the speed of the region where the links are located, is more accurate as compared to $v_2$. 
We further note that the convergence time of the presented speed estimation results is typically within a couple of minutes, with several cases (those with higher speeds) converging in much less than a minute.
\begin{figure*}
\begin{minipage}[c]{0.3\textwidth}
\includegraphics[width=1\linewidth]{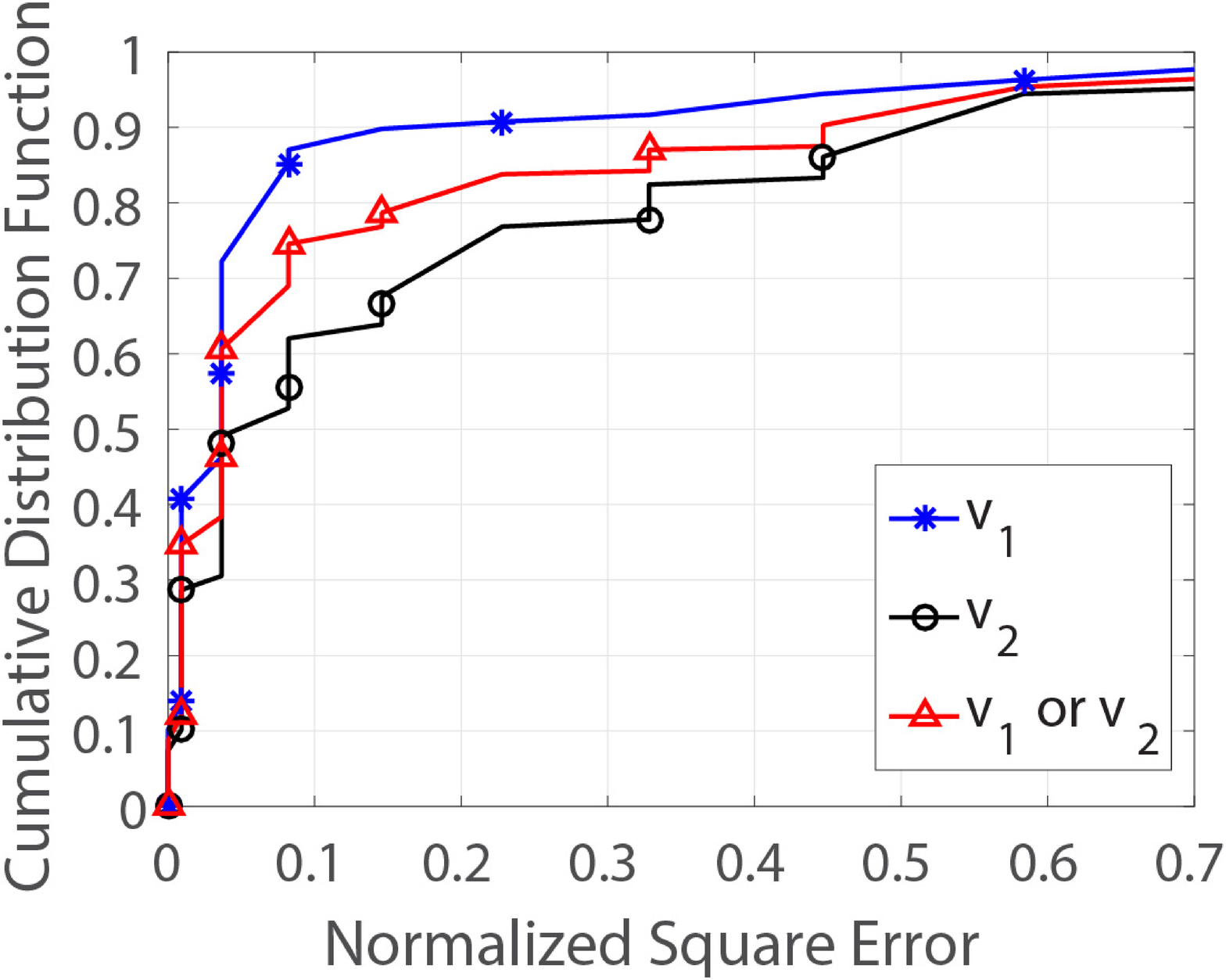}
\caption{CDF of the normalized square error for speeds in Region 1 ($v_1$), Region 2 ($v_2$), and for the speeds in any region. It can be seen that our approach estimates the speeds with a good accuracy.}
\label{fig_cdf_plot_all}
\end{minipage}
\hspace{0.15in}
\begin{minipage}[c]{0.3\textwidth}	
\includegraphics[width=0.95\linewidth]{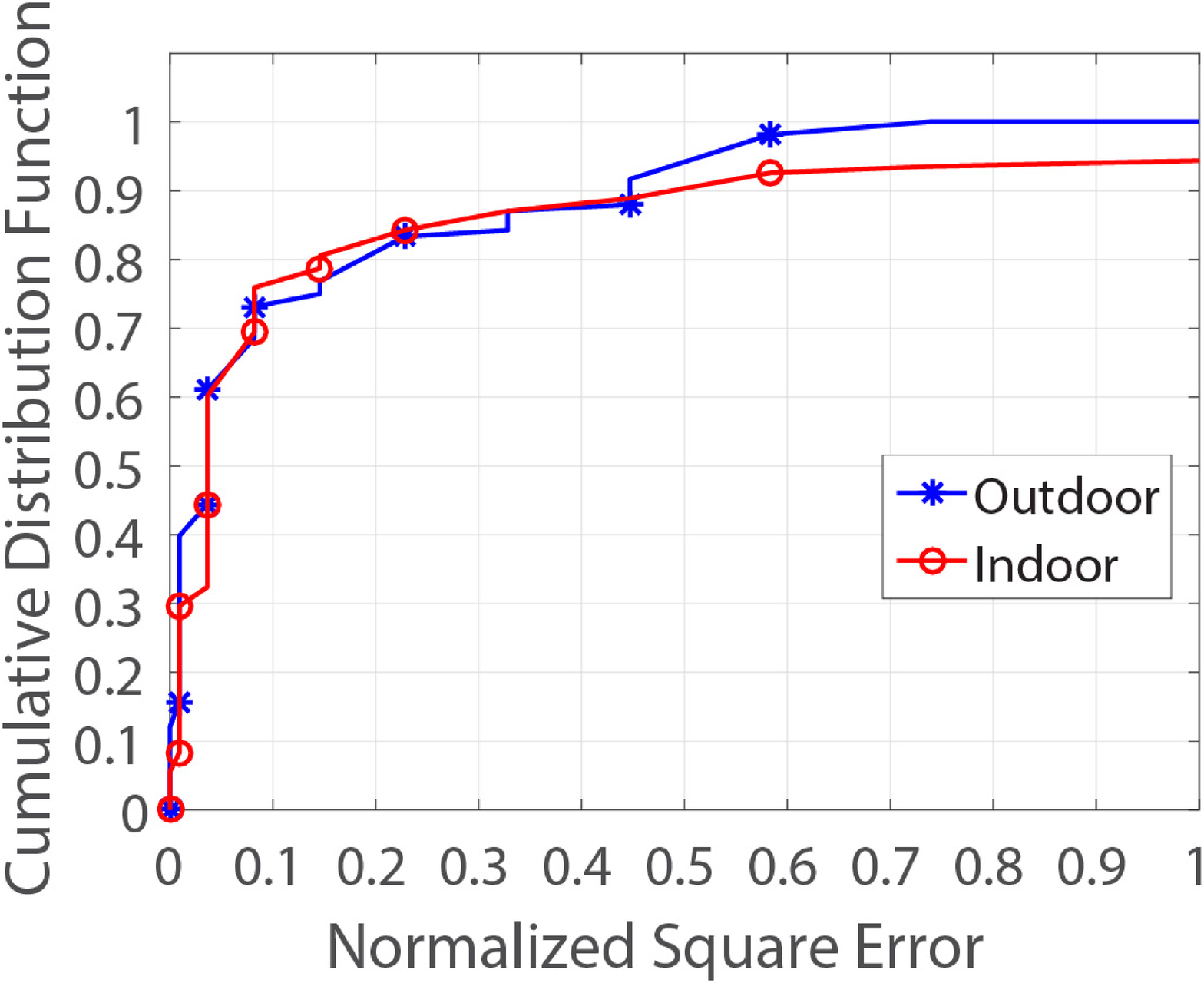}
\caption{CDF of the normalized square error based on the location of the experiment. It can be seen that the outdoor location has a slightly better performance than indoor, as expected.}
\label{fig_cdf_plot_region}
\end{minipage}
\hspace{0.2in}
\begin{minipage}[c]{0.3\textwidth}
\includegraphics[width=1\linewidth]{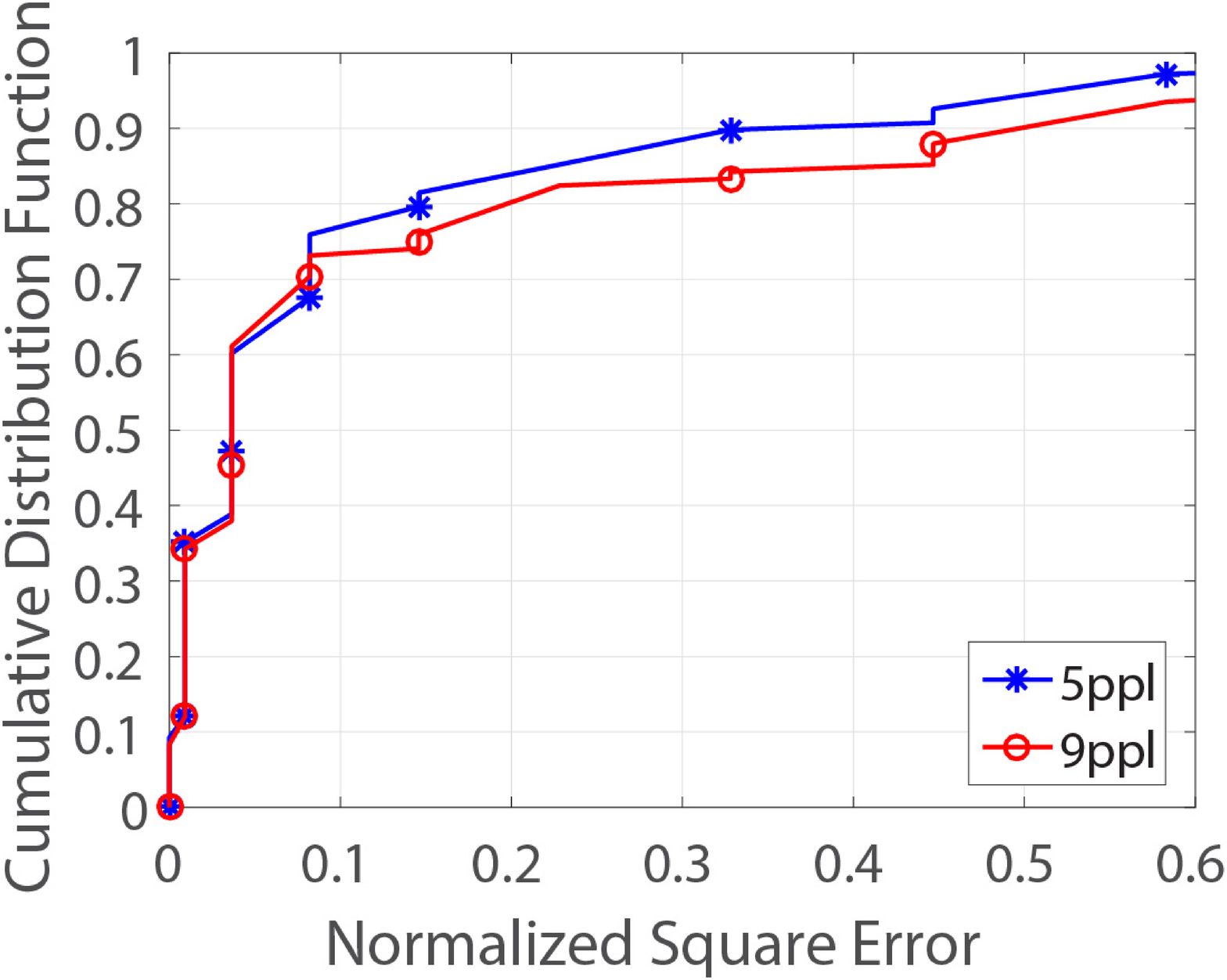}
\caption{CDF of the normalized square error based on the total number of people. It can be seen that the estimation error slightly increases for $9$ people as compared to the case of $N=5$.}\label{fig_cdf_plot_people}
\label{fig_cdf_plot}
\end{minipage}
\end{figure*}
\subsubsection{Speed Classification Performance}
Thus far, we have established that our approach can successfully estimate the region-dependent speeds of people walking in two adjacent regions, based on WiFi RSSI measurements in only one region. However, for some applications, an exact speed estimation may not be necessary. Rather, a classification of the pace to low, normal walking, or high may suffice. Therefore, we next show the classification performance of the proposed approach to Low ($0.3\ m/s$), Normal walking ($0.8\ m/s$), or High ($1.6\ m/s$) speeds. More specifically, we classify the estimated speed $\widehat{v_i}$ using nearest neighbor classifier 
as Low if $\widehat{v_i} \leq 0.55\  m/s$,  Normal if $0.55 \  m/s < \widehat{v_i} \leq 1.2\  m/s$, and High if $\widehat{v_i} > 1.2\  m/s$, for $i \in \{1,2\}$. Table \ref{tab_class_accuracy} shows the accuracy of our classification for both indoor and outdoor cases and for different number of people. It can be seen that the overall classification accuracy of the speeds in either of the two regions is $85.2\%$ over all the experiments, confirming a good performance. For comparison, we note that the probability of correct classification would have been $33\%$ in any of the two regions for a random classifier.
\begin{table}[t]
\renewcommand{\arraystretch}{2.15}%
\scalebox{1}
{
\begin{tabular}{|C{2 cm}|C{1.5 cm}|C{1.5 cm}|c|}
\cline{2-4}
\multicolumn{1}{c|}{} & \multicolumn{3}{c|}{\textbf{Classification accuracy (in \%)}}  \\
\hline
\shortstack{\textbf{Experiment} \\ \textbf{scenario}} & \shortstack{$v_1$} & \shortstack{$v_2$} &\shortstack{ ${v_1}$ or $v_2$}   \\ \hline
\shortstack{\bf Outdoor \\\bf N=5 people }&100& 81.4&90.7\\ \hline
\shortstack{\bf Outdoor \\ \bf N=9 people}&88.9&77.8&83.4\\ \hline
\shortstack{\bf Indoor \\ \bf N=5 people}&100&66.7&83.3\\ \hline
\shortstack{\bf Indoor \\ \bf N=9 people}&92.6&74.1&83.3\\ \hline
\shortstack{\bf All \\ \bf experiments}&95.4&75& 85.2\\ \hline
\end{tabular}
}		
\caption{Performance of speed classification to High, Normal Walking, and Low for indoor and outdoor cases, and for different number of pedestrians. }\label{tab_class_accuracy}
\end{table}
\subsubsection{Underlying Trends of Speed Estimation}
We next discuss some of the underlying characteristics of the results, starting with the impact of the experiment location.  Table \ref{tab_nmse_region} and Fig.~\ref{fig_cdf_plot_region} show the NMSE of the estimation error and the CDF of the normalized square error respectively, based on all the experiments in each location. While the estimation error in the indoor environment is still small, the estimation error is less in the outdoor environment as expected, due to the smaller amount of multipath. Furthermore, Table \ref{tab_nmse_people} and Fig.~\ref{fig_cdf_plot_people} show the performance as a function of the total number of pedestrians.  It can be seen that the estimation error is slightly higher for $N=9$ people as compared to $N=5$. 
\subsubsection{Sensitivity to $\theta_\textnormal{max}$}
As described in Section \ref{sec_speed_estimation}, we assume $\theta_\textnormal{max} = 45\degree$ in our models of the closed area, which characterizes the flow of people in hallway-type scenarios. We next show the sensitivity of our results to the assumed value of $\theta_\textnormal{max}$. More specifically, we assume a broad range of values for $\theta_\textnormal{max}$ to characterize the flow of people in our experiments and estimate the speeds of people accordingly. Fig. \ref{fig_sensitivity} shows the NMSE of the estimated speeds in the two regions as a function of the assumed value of $\theta_\textnormal{max}$. It can be seen that the estimation error is nearly constant over a broad range of $\theta_\textnormal{max}$, which shows that our approach is robust and not that sensitive to the exact choice of $\theta_\textnormal{max}$.
\begin{figure}[t]
\centering
\includegraphics[width=0.65\linewidth]{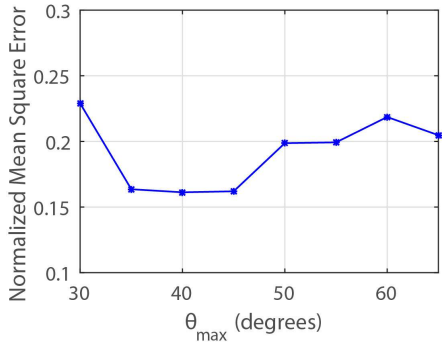}
\caption{Effect of the assumed value of $\theta_\textnormal{max}$ on the Normalized Mean Square Error of the estimated speeds in the two regions. It can be seen that NMSE is low for a broad range of $\theta_\textnormal{max}$, which shows that it is not that sensitive to the exact choice of $\theta_\textnormal{max}$}.
\label{fig_sensitivity}
\end{figure}
\subsection{Museum Experiments}
So far, we presented our experimental results for several cases in which people are walking with a variety of speeds in two adjacent regions of an area. We next consider a museum-type scenario, in which there are two adjacent exhibitions, showcasing two very different types of displays. We then utilize our methodology to estimate the visitor speeds in both exhibits, and deduce which exhibit is more popular. By more popular, we mean that the exhibit received more attention, i.e., people slowed down to spend more time there.
\begin{figure}[t]
\centering
\includegraphics[width=0.9\linewidth]{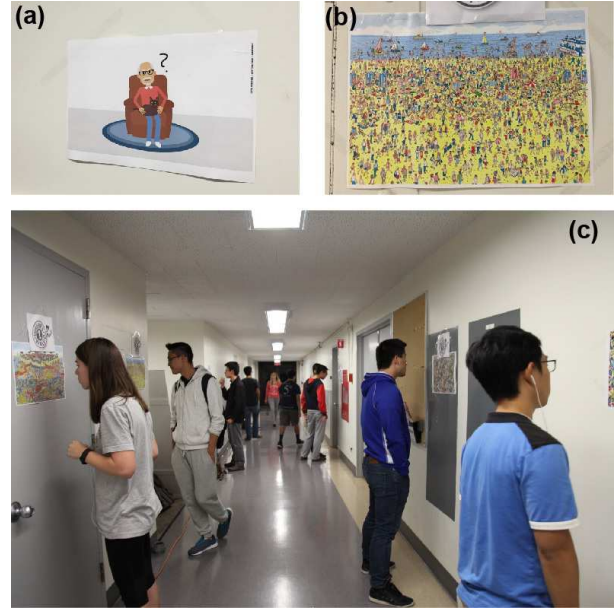}
\caption{Our museum which contains two exhibits -- (a) a sample display in the exhibit of Region 1, which contains non-engaging items, (b) a sample display in the exhibit of Region 2, which contains more engaging displays such as ``Where is Waldo?", and (c) a snapshot of the visitors exploring the museum.}
\label{fig_waldo_scenario}
\end{figure}
For the purpose of this experiment, we stage an exhibition with two types of exhibits in two adjacent regions. We place basic visually-boring displays on the walls of Region 1, such as basic pictures, list of alphabets, and list of numbers. In Region 2, on the other hand,  we place more visually-involved displays such as ``Where is Waldo"  pictures \cite{waldoref}. Fig.~\ref{fig_waldo_scenario} (a) and (b) show a sample display in Region 1 and Region 2 respectively. We use the indoor experiment site shown in Fig.~\ref{fig_indoor_scenario} for this experiment. We then invite $10$ people (randomly selected from our advertisement) to visit this museum. The visitors do not have any background about our experiments. Upon arrival, they are told to explore the area that consists of the two exhibits as it interests them.  Note that we do not ask people to walk at a particular speed in a given region, as we did in the validation experiments.
Fig.~\ref{fig_waldo_scenario} (c) shows a snapshot of the museum with people exploring the exhibits. We use the same Tx/Rx locations in Region 1 of Fig.~\ref{fig_indoor_scenario} and collect the data for 5 minutes. In this setting, we observe that people stop at a display that interests them before moving on to explore other displays. The experiment is videotaped in both regions and the ground-truth average speeds of people in Region 1 and Region 2 are visually estimated as $1.1\ m/s$  and $0.12\ m/s$, respectively, by extracting the time spent by each person in the two regions from the video. We then use our proposed approach to estimate the average speeds in the two regions of the museum. Fig.~\ref{fig_waldo_result} shows the estimated average speeds in the two regions as a function of time. It can be seen that the speed of people in the Exhibit of Region 2, which contains the Waldo pictures, is estimated as $0.3\ m/s$, indicating a significant slow down, while the speed in the Exhibit of Region 1 is estimated as $1\ m/s$, which is a normal walking speed. It can be seen that these estimates are consistent with the ground-truth and what one would expect based on the level of engagement of the displays. The estimates further indicate that Exhibit 2 was more engaging and popular since it was estimated that people significantly slowed down there. This shows the potential of the proposed methodology for estimating the level of popularity of adjacent displays, based on only sensing and measurement in one of the regions.

\begin{figure}[t]
\centering
\includegraphics[width=0.76\linewidth]{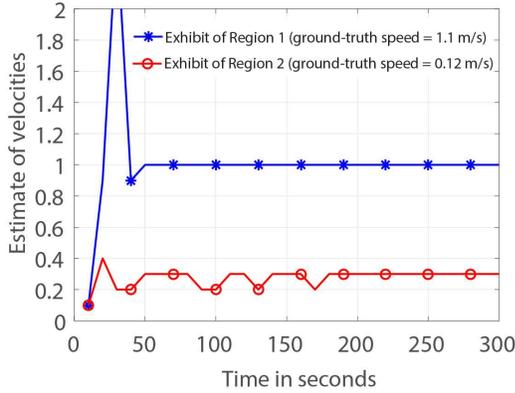}
\caption{Our estimates of the speeds in the two exhibits of the museum experiment of Fig.~\ref{fig_waldo_scenario}. The speed in the Exhibit of Region 2, which contains the Waldo pictures, is estimated as $ 0.3\ m/s$, indicating a significant slow down, while the speed in Region 1 is estimated as $ 1\ m/s$, which is a normal walking speed. The results further indicate that the exhibit of Region 2 was more engaging and popular.}
\label{fig_waldo_result}
\end{figure}
\subsection{Costco Experiments}
In this section, we use our framework to estimate the motion behavior of the buyers in an aisle of a retail store, Costco \cite{depatla2018secon}. Since people constantly come and go through the aisle, this will be an example of the open area scenario of Fig.~\ref{fig_illustration}b. Since the aisle that we were assigned by the store for our experiments only contained one kind of products, we then estimate the rate of arrival of people into the aisle, and the speed at which people walk while they are exploring the aisle (using the same framework), thus assessing the popularity of the products in the aisle.

Fig.~\ref{fig_costco} shows the aisle of interest in our local Costco. This aisle contains a specific type of merchandise, snacks and cookies in this case. Both ends of the aisle are open and people can enter/exit from either end of the aisle. 
Since the aisle contains the same type of products, we take the entire aisle as a single region (i.e., $v_1=v_2$), but assume the rate of arrival (or equivalently $N_\textnormal{avg}$) to be unknown as well.
It is expected that people walk at a slow pace if the products in the aisle generate interest and they consider buying them. We are thus interested in estimating such behaviors. A pair of WiFi links are located along the aisle, as indicated in Fig.~\ref{fig_costco}, and make wireless measurements as people walk through the aisle. We then use our approach of Section \ref{sec_speed_estimation} to estimate the speed of people in the aisle as well as their rate of arrival into the aisle. 

Since the probability of crossing link $i$, $p_{c,i}= \lambda \delta t, \textnormal{\ for \ } i \in \{1,2\}$, the rate of arrival $\lambda$ is estimated as
$\widehat{\lambda} = \frac{p_{c,1} + p_{c,2}}{2\delta t}$.
In order to estimate the speed of people walking in the aisle, we further use the cross-correlation between the two WiFi links given by Eq.~(\ref{eq_cross_corr}).



\end{section}
\begin{figure}[t]
\centering
\includegraphics[width=0.9\linewidth]{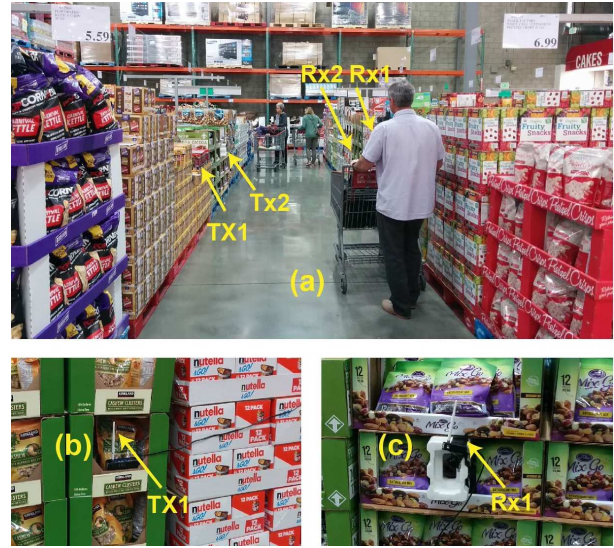}
\caption{The Costco experiment -- (a) shows the considered ``snacks and cookies" aisle in Costco, while (b) and (c) show a pair of our WiFi nodes positioned along the aisle to make wireless measurements.}\label{fig_costco}
\end{figure}
\begin{figure}[t]
\centering
\includegraphics[width=0.8\linewidth]{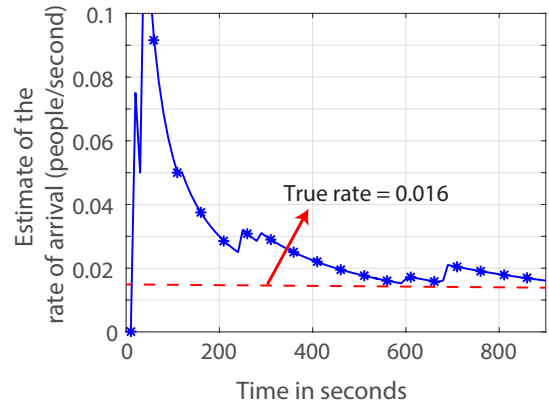}
\vspace{-0.05in}
\caption{The estimate of the rate of arrival of people into the aisle of Fig.~\ref{fig_costco} at Costco, as a function of time. It can be seen that our framework correctly estimates the rate of arrival.}\label{fig_rate_func_time}
\end{figure}
We then collect wireless RSSI measurements for $15$ minutes as people walk through the aisle shown in Fig. \ref{fig_costco}. We manually record the times at which people arrive from either entrance of the aisle 
and compute the true rate of arrival. Fig.~\ref{fig_rate_func_time} shows the estimated rate of arrival as a function of time. It can be seen that our framework accurately estimates the rate of arrival of people into the aisle using a pair of WiFi links. Note that the rate of arrival on that particular day/time was
$1$ person per minute. Thus, our estimation converges relatively fast, within $400$ seconds, which is the time $6$ people visited the aisle. Furthermore, the average ground-truth speed of people walking in that aisle is estimated as $0.48\ m/s$, by manually recording the entrance and exit times of people in that aisle on $4$ different days. The average speed of people walking in the aisle is estimated as $0.2\ m/s$ using our framework, which is consistent with the ground-truth, and indicates a significant slow down, showcasing the popularity of the aisle. 

\begin{section}{Conclusion}\label{sec_conclusions}
In this paper, we proposed a framework to estimate the average speeds of pedestrians in two adjacent regions, by using RSSI measurements of a pair of WiFi links in only one region. Our approach only relies on WiFi signal availability in the region where the links are located. Thus, it not only allows for estimating the speed of a crowd in the immediate region where the pair of links are, but also enables deducing the speed of the crowd in the adjacent WiFi-free regions. More specifically, we showed how two key statistics, the probability of crossing and the cross-correlation between the two links, carry key information about the pedestrian speeds in the two regions and mathematically characterized them as a function of the speeds. To validate our framework, we ran extensive experiments (total of 108) in indoor and outdoor locations with up to $10$ people, with a variety of speeds per region, and showed that our approach can accurately estimate the speeds of pedestrians in both regions. Furthermore, we tested our methodology in a museum setting, with two different exhibitions in adjacent areas, and estimated the average pedestrian speeds in both exhibits, thus deducing which exhibit was more popular. Finally, we used our framework in Costco, estimated the motion behavior of buyers in an aisle, and deduced the popularity of the products located in that aisle.
\end{section}




\bibliographystyle{IEEEtran}
\bibliography{ref_iot}

\begin{IEEEbiography}[{\includegraphics[width=0.95in,height=1.25in,clip]{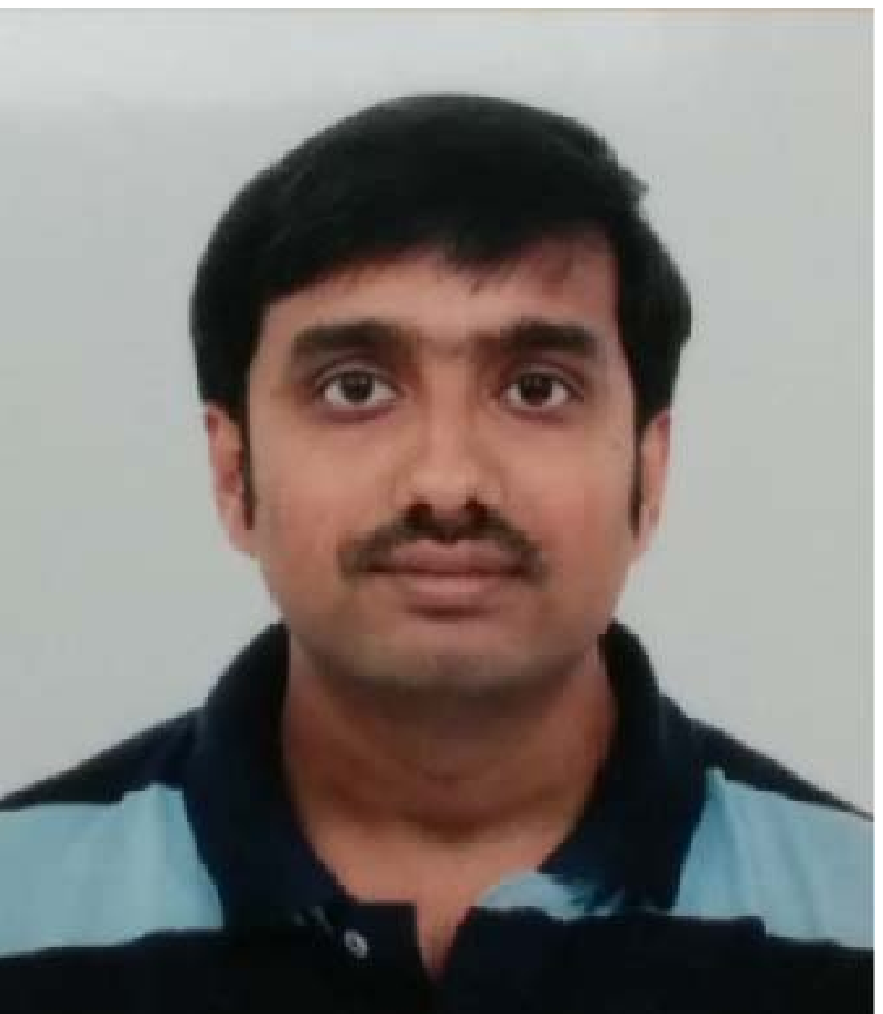}}]{Saandeep Depatla}
received the B.S. degree in electronics and communication engineering from the National Institute of Technology, Warangal, in 2010 and the M.S. degree in electrical and computer science engineering (ECE) from the University of California, Santa Barbara (UCSB), in 2014. From 2010 to 2012, he worked on developing antennas for radars in electronics and radar development establishment, India. Since 2013, he has been working towards the Ph.D. degree in ECE at UCSB. His research interests include ambient sensing using wireless signals, signal processing, and wireless communications. His research on RF sensing has appeared in several popular news venues such as Engadget and Huffington Post among others. 
\end{IEEEbiography}
\vspace{-5.5in}
\begin{IEEEbiography}[{\includegraphics[width=1in,height=1.25in,clip,keepaspectratio]{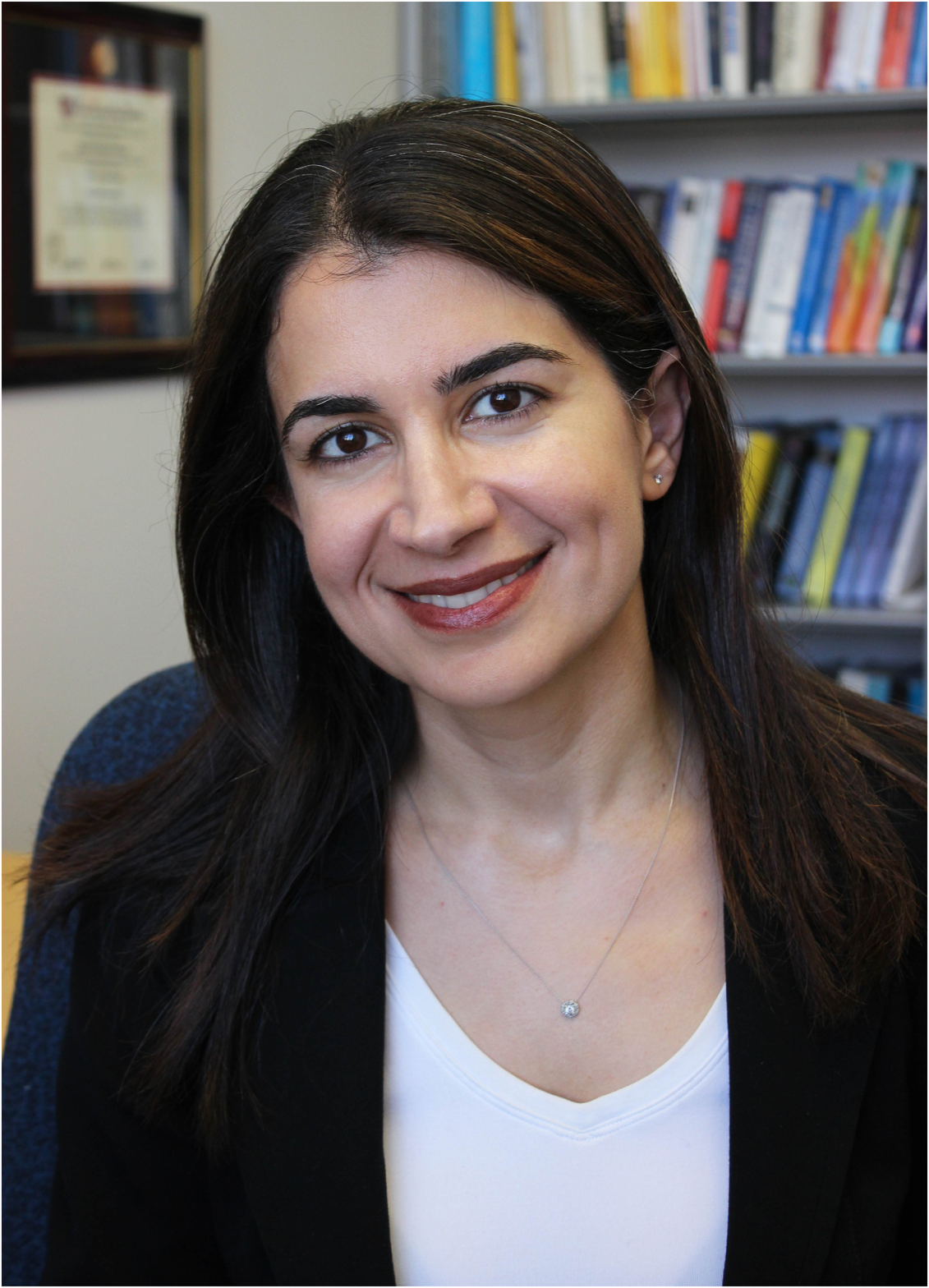}}]{Yasamin Mostofi}
	 received the B.S. degree in electrical engineering from Sharif University of Technology, and the M.S. and Ph.D. degrees from Stanford University. She is currently a professor in the Department of Electrical and Computer Engineering at the University of California Santa Barbara. Yasamin is the recipient of the 2016 Antonio Ruberti Prize from the IEEE Control Systems Society, the Presidential Early Career Award for Scientists and Engineers (PECASE), the National Science Foundation (NSF) CAREER award, and the IEEE 2012 Outstanding Engineer Award of Region 6 (more than 10 Western U.S. states), among other awards. Her research is at the intersection of communications and robotics, on mobile sensor networks. Current research thrusts include X-ray vision for robots, RF sensing, communication-aware robotics, occupancy estimation, see-through imaging, and human-robot networks. Her research has appeared in several reputable news venues such as BBC, Huffington Post, Daily Mail, Engadget, TechCrunch, NSF Science360, ACM News, and IEEE Spectrum, among others. 
\end{IEEEbiography}

\end{document}